\documentclass[conference]{svproc}
\usepackage{times}

\usepackage{amsmath,amssymb,mathrsfs}
\usepackage{enumitem}
\usepackage{scalerel,stackengine}
\usepackage[usenames,dvipsnames]{xcolor}
\usepackage{cite}
\usepackage{mdframed}
\usepackage{algpseudocode}
\usepackage[font=footnotesize]{caption}
\usepackage{algorithm}
\usepackage{graphics} % for pdf, bitmapped graphics files
\usepackage{subcaption}
\usepackage[title]{appendix}

\newmdenv[topline=false,bottomline=false,rightline=false]{leftbox}

\newtheorem{assumption}{Assumption}

\stackMath
\newcommand\wwidehat[1]{%
\savestack{\tmpbox}{\stretchto{%
  \scaleto{%
    \scalerel*[\widthof{\ensuremath{#1}}]{\kern-.6pt\bigwedge\kern-.6pt}%
    {\rule[-\textheight/2]{1ex}{\textheight}}%WIDTH-LIMITED BIG WEDGE
  }{\textheight}% 
}{0.5ex}}%
\stackon[1pt]{#1}{\tmpbox}%
}

\newcommand{\revision}[1]{{\color{black}{#1}}}

\newcommand{\argmin}{\operatornamewithlimits{argmin}}

\newcommand{\X}{\mathcal{X}}
\newcommand{\Y}{\reals^n}
\newcommand{\Hk}{\mathcal{H}_{K}}
\newcommand{\Lin}{\mathcal{L}}
\newcommand{\reals}{\mathbb{R}}
\newcommand{\ip}[2]{\left\langle #1, #2 \right\rangle}

\newcommand{\V}{\mathcal{V}}
\newcommand{\Kb}{K^B}
\newcommand{\Hkb}{\mathcal{H}_K^B}
\newcommand{\Vb}{\mathcal{V}_{B}}
\newcommand{\Vf}{\mathcal{V}_{f}}

\newcommand{\bs}{\mathfrak{b}}

\newcommand{\kwp}{\hat{\kappa}}
\newcommand{\kw}{\kappa}
\newcommand{\Hwp}{\mathcal{H}_{\hat{\kappa}}}
\newcommand{\Hw}{\mathcal{H}_{\kappa}}

\newcommand{\Sj}{\mathbb{S}}
\newcommand{\Sjpp}{\mathbb{S}^{>0}}

\newcommand{\wl}{\underline{w}}
\newcommand{\wu}{\overline{w}}

\newcommand{\xs}{x_i}
\newcommand{\us}{u_i}

\newcommand{\dx}{\delta_x}
\newcommand{\ddx}{\dot{\delta}_x}

\renewcommand{\baselinestretch}{0.9}

\title{Learning Stabilizable Dynamical Systems\\ via Control Contraction Metrics}

\author{Sumeet Singh\inst{1} \and Vikas Sindhwani\inst{2}\and Jean-Jacques E. Slotine\inst{3}\and Marco Pavone\inst{1}
\thanks{This work was supported by NASA under the Space Technology Research Grants Program, Grant NNX12AQ43G, and by the King Abdulaziz City for Science and Technology (KACST).}
}

\institute{Dept. of Aeronautics and Astronautics, Stanford University \\ \texttt{\{ssingh19,pavone\}@stanford.edu}
\and
Google Brain Robotics, New York \\ \texttt{sindhwani@google.com}
\and
Dept. of Mechanical Engineering, Massachusetts Institute of Technology \\ \texttt{jjs@mit.edu}}

\begin{document}
\maketitle

%===============================================================================

\vspace{-6mm}
\begin{abstract}
We propose a novel  framework for learning stabilizable nonlinear dynamical systems for continuous control tasks in robotics. The key idea is to develop a new control-theoretic regularizer for dynamics fitting rooted in the notion of {\it stabilizability}, which guarantees that the learned system can be accompanied by a robust controller capable of stabilizing {\it any} open-loop trajectory that the system may generate. By leveraging tools from contraction theory, statistical learning, and  convex optimization, we provide a general and tractable \revision{semi-supervised} algorithm to learn stabilizable dynamics, which can be applied to complex underactuated systems. We validated the proposed algorithm on a simulated planar quadrotor system and observed \revision{notably improved trajectory generation and tracking performance with the control-theoretic regularized model over models learned using traditional regression techniques, especially when using a small number of demonstration examples}. The results presented illustrate the need to infuse standard model-based reinforcement learning algorithms with concepts drawn from nonlinear control theory for improved reliability. 
\end{abstract}
\vspace{-6mm}

% Two or three meaningful keywords should be added here

%===============================================================================

\section{Introduction}

The problem of efficiently and accurately estimating an unknown dynamical system, \begin{equation}
    \dot{x}(t) = f(x(t),u(t)), 
\label{ode}
\end{equation} from a small set of sampled trajectories, where $x \in \reals^n$ is the state and $u \in \reals^m$ is the control input, is the central task in model-based Reinforcement Learning (RL). In this setting, a robotic agent strives to pair an estimated  dynamics model with a feedback policy in order to optimally act in a dynamic and uncertain environment.  The model of the dynamical system can be continuously updated as the robot experiences the consequences of its actions, and the improved model can be  leveraged for different tasks affording a natural form of transfer learning. When it works, model-based Reinforcement Learning typically offers major improvements in sample efficiency in comparison to state of the art RL methods such as Policy Gradients~\cite{ChuaCalandraEtAl2018,NagabandiKahnEtAl2017} that do not explicitly estimate the underlying system. Yet, all too often, when standard supervised learning with powerful function approximators such as Deep Neural Networks and Kernel Methods are applied to model complex dynamics, the resulting controllers do not perform at par with model-free RL methods in the limit of increasing sample size, due to compounding errors across long time horizons. The main goal of this paper is to develop a new control-theoretic regularizer for dynamics fitting rooted in the notion of {\it stabilizability}, which guarantees that the learned system can be accompanied by a robust controller capable of stabilizing any trajectory that the system may generate.

%===============================================================================

%\section{Problem Statement}
%
\iffalse
Consider a robotic system whose dynamics are described by the generic nonlinear differential equation
\begin{equation}
    \dot{x}(t) = f(x(t),u(t)), 
\label{ode}
\end{equation}
where $x \in \reals^n$ is the state, $u \in \reals^m$ is the control input. We assume that the function $f$ is smooth. A state-input trajectory satisfying~\eqref{ode} is denoted as the pair $(x,u)$. The key concept leveraged in this work is the notion of \emph{stabilizability}. 
\fi
Formally, a reference state-input trajectory pair $(x^*(t), u^*(t)),\ t \in [0,T]$ for system~\eqref{ode} is termed \emph{exponentially stabilizable at rate $\lambda>0$} if there exists a feedback controller $k : \reals^n \times \reals^n \rightarrow \reals^m$ such that the solution $x(t)$ of the system:
\[
    \dot{x}(t) = f(x(t), u^*(t) + k(x^*(t),x(t))),
\]
converges exponentially to $x^*(t)$ at rate $\lambda$. That is,
\begin{equation}
    \|x(t) - x^*(t)\|_2 \leq C \|x(0) - x^*(0)\|_2 \ e^{-\lambda t}
\label{exp_stab}
\end{equation}
for some constant $C>0$. The \emph{system}~\eqref{ode} is termed \emph{exponentially stabilizable at rate $\lambda$} in an open, connected, bounded region $\X \subset \reals^n$ if all state trajectories $x^*(t)$ satisfying $x^*(t) \in \X,\ \forall t \in [0,T]$ are exponentially stabilizable at rate $\lambda$. 

%\ssmargin{relax stabilizability to boundedness - Lyapunov stable, or asymptotic stable? Leave exponential stability to when we talk about contraction}{}

{\bf Problem Statement}: In this work, we assume that the dynamics function $f(x,u)$ is unknown to us and we are instead provided with a dataset of tuples $\{(\xs, \us, \dot{x}_i)\}_{i=1}^{N}$ taken from a collection of observed trajectories (e.g., expert demonstrations) on the robot. Our objective is to solve the problem:
\begin{align}
    \min_{\hat{f} \in \mathcal{H}} \quad & \sum_{i=1}^{N} \left\| \hat{f}(\xs,\us) - \dot{x}_i \right\|_2^2 + \mu \|\hat{f}\|^2_{\mathcal{H}} \label{prob_gen} \\
    \text{s.t.} \quad & \text{$\hat{f}$ is stabilizable,}
\end{align}
where $\mathcal{H}$ is an appropriate normed function space and $\mu >0$ is a regularization parameter. Note that we use $(\hat{\cdot})$ to differentiate the learned dynamics from the true dynamics. We expect that for systems that are indeed stabilizable, enforcing such a constraint may drastically \emph{prune the hypothesis space, thereby playing the role of a ``control-theoretic'' regularizer} that is potentially more powerful and ultimately, more pertinent for the downstream control task of generating and tracking new trajectories.
 
{\bf Related Work}:  The simplest approach to learning dynamics is to ignore stabilizability and treat the problem as a standard one-step time series regression task~\cite{NagabandiKahnEtAl2017,ChuaCalandraEtAl2018,DeisenrothRasmussen2011}. However, coarse dynamics models trained on limited training data typically generate trajectories that rapidly diverge from expected paths, inducing controllers that are ineffective when applied to the true system. This divergence can be reduced by expanding the training data with corrections to boost multi-step prediction accuracy~\cite{VenkatramanHebertEtAl2015, VenkatramanCapobiancoEtAl2016}. In recent work on uncertainty-aware model-based RL, policies~\cite{NagabandiKahnEtAl2017,ChuaCalandraEtAl2018} are optimized with respect to stochastic rollouts from probabilistic dynamics models that are iteratively improved in a model predictive control loop. Despite being effective, these methods are still heuristic in the sense that the existence of a stabilizing feedback controller is not explicitly guaranteed. 

Learning dynamical systems satisfying some desirable stability properties (such as asymptotic stability about an equilibrium point, e.g., for point-to-point motion) has been studied in the autonomous case, $\dot{x}(t) = f(x(t))$, in the context of imitation learning. In this line of work, one assumes perfect knowledge and invertibility of the robot's \emph{controlled} dynamics to solve for the input that realizes this desirable closed-loop motion~\cite{LemmeNeumannEtAl2014,Khansari-ZadehKhatib2017,SindhwaniTuEtAl2018,RavichandarSalehiEtAl2017,Khansari-ZadehBillard2011,MedinaBillard2017}. Crucially, in our work, we \emph{do not} require knowledge, or invertibility of the robot's controlled dynamics. We seek to learn the full controlled dynamics of the robot, under the constraint that the resulting learned dynamics generate dynamically feasible, and most importantly, stabilizable trajectories. Thus, this work generalizes existing literature by additionally incorporating the controllability limitations of the robot within the learning problem. In that sense, it is in the spirit of recent model-based RL techniques that exploit control theoretic notions of stability to guarantee model safety during the learning process~\cite{BerkenkampTurchettaEtAl2017}. However, unlike the work in~\cite{BerkenkampTurchettaEtAl2017} which aims to maintain a local region of attraction near a known safe operating point, we consider a stronger notion of safety -- that of stabilizability, that is, the ability to keep the system within a bounded region of any exploratory open-loop trajectory.

Potentially, the tools we develop may also be used to extend standard adaptive robot control design, such as~\cite{SlotineLi1987} -- a technique which achieves stable concurrent learning and control using a combination of physical basis functions and general mathematical expansions, e.g. radial basis function approximations~\cite{SannerSlotine1992}. Notably, our work allows us to handle complex underactuated systems, a consequence of the significantly more powerful function approximation framework developed herein, as well as of the use of 
a differential (rather than classical) Lyapunov-like setting, as we shall detail.

%{\color{red} Although we have to say something about adaptive control, this is actually a rather
%separate point, as adaptive control does not assume measurement of $\dot{x} \ $.}
%{\color{blue} Sumeet: I will modify this point; will re-phrase discussion on adaptive as a separate point, particularly in context of underactuated systems}

%{\color{red} Also, we should qualify a little what we mean
%by RL, in robotics it evokes e.g. the classical work of
%Kenji Doya where a humanoid robot leanrs to stand up by itself.}
%{\color{blue} See above points.
%}

{\bf Statement of Contributions:} Stabilizability of trajectories is a complex task in non-linear control. In this work, we leverage recent advances in contraction theory for control design through the use of \emph{control contraction metrics} (CCM)~\cite{ManchesterSlotine2017} that turns stabilizability constraints into convex Linear Matrix Inequalities (LMIs). Contraction theory~\cite{LohmillerSlotine1998} is a method of analyzing nonlinear systems in a differential framework, i.e., via the associated variational system~\cite[Chp 3]{CrouchSchaft1987}, and is focused on the study of convergence between pairs of state trajectories towards each other. Thus, at its core, contraction explores a stronger notion of stability -- that of incremental stability between solution trajectories, instead of the stability of an equilibrium point or invariant set. Importantly, we harness recent results in~\cite{ManchesterTangEtAl2015,ManchesterSlotine2017,SinghMajumdarEtAl2017} that illustrate how to use contraction theory to obtain a \emph{certificate} for trajectory stabilizability and an accompanying tracking controller with exponential stability properties. In Section~\ref{sec:ccms}, we provide a brief summary of these results, which in turn will form the foundation of this work.
 
 Our paper makes four primary contributions. First, we formulate the learning stabilizable dynamics problem through the lens of control contraction metrics (Section~\ref{sec:prob}). Second, under an arguably weak assumption on the sparsity of the true dynamics model, we present a finite-dimensional optimization-based solution to this problem by leveraging the powerful framework of vector-valued Reproducing Kernel Hilbert Spaces (Section~\ref{sec:finite}). We further motivate this solution from a standpoint of viewing the stabilizability constraint as a novel control-theoretic \emph{regularizer} for dynamics learning. Third, we develop a tractable algorithm leveraging alternating convex optimization problems and adaptive sampling to iteratively solve the finite-dimensional optimization problem (Section~\ref{sec:soln}). Finally, we verify the proposed approach on a 6-state, 2-input planar quadrotor model where we demonstrate that naive regression-based dynamics learning can yield estimated models that \revision{generate completely unstabilizable trajectories}. In contrast, \revision{the control-theoretic regularized model generates vastly superior quality, trackable trajectories, especially} for smaller training sets (Section~\ref{sec:result}).

%\ssmargin{add the following: Blocher contraction (learning autonomous systems with a correction term to ensure contraction holds. correction term smoothly modulated to go to 0 near demonstrations and in full effect away from demonstrations.}{}
\vspace{-2mm}
\section{Review of Contraction Theory} \label{sec:ccms}
\vspace{-2mm}

The core principle behind contraction theory~\cite{LohmillerSlotine1998} is to study the evolution of distance between any two \emph{arbitrarily close} neighboring trajectories and drawing conclusions on the distance between \emph{any} pair of trajectories.  Given an autonomous system of the form: $\dot{x}(t) = f(x(t))$, consider two neighboring trajectories separated by an infinitesimal (virtual) displacement $\delta_x$ (formally, $\delta_x$ is a vector in the tangent space $\mathcal{T}_x \X$ at $x$). The dynamics of this virtual displacement are given by:
\[
    \dot{\delta}_x = \dfrac{\partial f}{\partial x} \delta_x,
\]
where $\partial f/\partial x$ is the Jacobian of $f$. The dynamics of the infinitesimal squared distance $\delta_x^T\delta_x$ between these two trajectories is then given by:
\[
    \dfrac{d}{dt}\left( \delta_x ^T \delta_x \right) = 2 \delta_x ^T \dfrac{\partial f}{\partial x} \delta_x.
\]
Then, if the (symmetric part) of the Jacobian matrix $\partial f/\partial x$ is \emph{uniformly} negative definite, i.e., 
\[
    \sup_{x} \lambda_{\max}\left(\dfrac{1}{2}\wwidehat{\dfrac{\partial f(x)}{\partial x}}\right) \leq -\lambda < 0,
\]
where $\wwidehat{(\cdot)} := (\cdot) + (\cdot)^T$, $\lambda > 0$, one has that the squared infinitesimal length $\delta_x^T\delta_x$ is exponentially convergent to zero at rate $2\lambda$. By path integration of $\delta_x$ between \emph{any} pair of trajectories, one has that the distance between any two trajectories shrinks exponentially to zero. The vector field is thereby referred to be \emph{contracting at rate $\lambda$}.

Contraction metrics generalize this observation by considering as infinitesimal squared length distance, a symmetric positive definite function $V(x,\delta_x) = \delta_x^T M(x)\delta_x$, where $M: \X \rightarrow \Sjpp_n$, is a mapping from $\X$ to the set of uniformly positive-definite $n\times n$ symmetric matrices. Formally, $M(x)$ may be interpreted as a Riemannian metric tensor, endowing the space $\X$ with the Riemannian squared length element $V(x,\delta_x)$. A fundamental result in contraction theory~\cite{LohmillerSlotine1998} is that \emph{any} contracting system admits a contraction metric $M(x)$ such that the associated function $V(x,\delta_x)$ satisfies:
\[
    \dot{V}(x,\delta_x) \leq - 2\lambda V(x,\delta_x), \quad \forall (x,\delta_x) \in \mathcal{T}\X,
\]
for some $\lambda >0$. Thus, the function $V(x,\delta_x)$ may be interpreted as a \emph{differential Lyapunov function}. 
\vspace{-2mm}
\subsection{Control Contraction Metrics}

Control contraction metrics (CCMs) generalize contraction analysis to the controlled dynamical setting, in the sense that the analysis searches \emph{jointly} for a controller design and the metric that describes the contraction properties of the resulting closed-loop system. Consider dynamics of the form:
\begin{equation}
    \dot{x}(t) = f(x(t)) + B(x(t)) u(t),
\label{dyn}
\end{equation}
where $B: \X \rightarrow \reals^{n\times m}$ is the input matrix, and denote $B$ in column form as $(b_1,\ldots,b_m)$ and $u$ in component form as $(u^1,\ldots,u^m)$. To define a CCM, analogously to the previous section, we first analyze the variational dynamics, i.e., the dynamics of an infinitesimal displacement $\delta_x$:
\begin{equation}
	\ddx= \overbrace{\bigg(\dfrac{\partial f(x)}{\partial x}  + \sum_{j=1}^m u^j \dfrac{\partial b_j(x)}{\partial x}\bigg)}^{:= A(x,u)}\delta_{x}+ B(x)\delta_{u},
\label{var_dyn_c}
\end{equation}
where $\delta_u$ is an infinitesimal (virtual) control vector at $u$ (i.e., $\delta_u$ is a vector in the control input tangent space, i.e., $\reals^m$). A CCM for the system $\{f,B\}$ is a uniformly positive-definite symmetric matrix function $M(x)$ such that there exists a function $\delta_u(x,\dx,u)$ so that the function $V(x,\dx) = \dx^T M(x) \dx$ satisfies
\begin{equation}
\begin{split}
    \dot{V}(x,\dx,u) &= \delta_{x}^{T}\left(\partial_{f+Bu}M(x)+ \wwidehat{M(x)A(x,u)} \right) \delta_{x} + 2 \delta_{x}^{T}M(x)B(x)\delta_{u} \\
    &\leq -2\lambda V(x,\dx), \quad \forall (x,\dx) \in \mathcal{T}\X,\ u \in \reals^m,
\end{split}
\label{V_dot}
\end{equation}
where $\partial_g M(x)$ is the matrix with element $(i,j)$ given by Lie derivative of $M_{ij}(x)$ along the vector $g$. Given the existence of a valid CCM, one then constructs a stabilizing (in the sense of eq.~\eqref{exp_stab}) feedback controller $k(x^*,x)$ as described in Appendix~\ref{ccm_appendix}.

Some important observations are in order. First, the function $V(x,\dx)$ may be interpreted as a differential \emph{control} Lyapunov function, in that, there exists a stabilizing differential controller $\delta_u$ that stabilizes the variational dynamics~\eqref{var_dyn_c} in the sense of eq.~\eqref{V_dot}. Second, and more importantly, we see that by stabilizing the variational dynamics (essentially an infinite family of linear dynamics in $(\delta_x,\delta_u)$) pointwise, everywhere in the state-space, we obtain a stabilizing controller for the original nonlinear system. Crucially, this is an exact stabilization result, not one based on local linearization-based control. Consequently, one can show several useful properties, such as invariance to state-space transformations~\cite{ManchesterSlotine2017} and robustness~\cite{SinghMajumdarEtAl2017,ManchesterSlotine2018}.  Third, the CCM approach only requires a weak form of controllability, and therefore is not restricted to feedback linearizable (i.e., invertible) systems. 

%===============================================================================
\vspace{-2mm}
\section{Problem Formulation}\label{sec:prob}
\vspace{-2mm}

Using the characterization of stabilizability using CCMs, we can now formalize our problem as follows. Given our dataset of tuples $\{(\xs,\us,\dot{x}_i)\}_{i=1}^{N}$, the objective of this work is to learn the dynamics functions $f(x)$ and $B(x)$ in eq.~\eqref{dyn}, subject to the constraint that there exists a valid CCM $M(x)$ for the learned dynamics. \revision{That is, the CCM $M(x)$ plays the role of a \emph{certificate} of stabilizability for the learned dynamics.}

As shown in~\cite{ManchesterSlotine2017}, a necessary and sufficient characterization of a CCM $M(x)$ is given in terms of its dual $W(x):= M(x)^{-1}$ by the following two conditions:
\begin{align}
	 B_{\perp}^{T}\left( \partial_{b_j}W(x) - \wwidehat{\dfrac{\partial b_j(x)}{\partial x}W(x)} \right)B_{\perp}= 0, \ j = 1,\ldots, m \quad &\forall x \in \X,
\label{killing_A} \\
	   \underbrace{B_{\perp}(x)^{T}\left(-\partial_{f}W(x) + \wwidehat{\dfrac{\partial f(x)}{\partial x}W(x)} + 2\lambda W(x) \right)B_{\perp}(x)}_{:=F(x;f,W,\lambda)} \prec 0, \quad &\forall x \in \X, \label{nat_contraction_W}
\end{align}
where $B_{\perp}$ is the annihilator matrix for $B$, i.e., $B(x)^T B_\perp(x) = 0$ for all $x$. In the definition above, we write $F(x;W,f,\lambda)$ since $\{W,f,\lambda\}$ will be optimization variables in our formulation. Thus, our learning task reduces to finding the functions $\{f,B,W\}$ and constant $\lambda$ that jointly satisfy the above constraints, while minimizing an appropriate regularized regression loss function. Formally, problem~\eqref{prob_gen} can be re-stated as: \vspace{-0.2cm}
\begin{subequations}\label{prob_gen2}
\begin{align}
&\min_{\substack{\hat{f} \in \mathcal{H}^{f}, \hat{b}_j \in \mathcal{H}^{B}, j =1,\ldots,m \\ W \in \mathcal{H}^W \\ \wl, \wu, \lambda \in \reals_{>0}}} && \overbrace{\sum_{i=1}^{N} \left\| \hat{f}(\xs) + \hat{B}(\xs) \us - \dot{x}_i \right\|_2^2  + \mu_f \| \hat{f} \|^2_{\mathcal{H}^f} + \mu_b \sum_{j=1}^{m} \| \hat{b}_j \|^2_{\mathcal{H}^B}}^{:= J_d(\hat{f},\hat{B})} + \nonumber \\
& \qquad && + \underbrace{(\wu-\wl) +  \mu_w \|W\|^2_{\mathcal{H}^W}}_{:=J_m(W,\wl,\wu)}  \\
&\qquad \text{subject to} && \text{eqs.~\eqref{killing_A},~\eqref{nat_contraction_W}} \quad \forall x \in \X, \\
& && \wl I_n \preceq W(x) \preceq \wu I_n, \quad \forall x \in \X, \label{W_unif}
\end{align}
\end{subequations}
where $\mathcal{H}^f$ and $\mathcal{H}^B$ are appropriately chosen $\Y$-valued function classes on $\X$ for $\hat{f}$ and $\hat{b}_j$ respectively, and $\mathcal{H}^W$ is a suitable $\Sjpp_n$-valued function space on $\X$. The objective is composed of a dynamics term $J_d$ -- consisting of regression loss and regularization terms, and a metric term $J_m$ -- consisting of a condition number surrogate loss on the metric $W(x)$ and a regularization term. The metric cost term $\wu-\wl$ is motivated by the observation that the state tracking error (i.e., $\|x(t)-x^*(t)\|_2$) in the presence of bounded additive disturbances is proportional to the ratio $\wu/\wl$ (see~\cite{SinghMajumdarEtAl2017}).

Notice that the coupling constraint~\eqref{nat_contraction_W} is a bi-linear matrix inequality in the decision variables sets $\{\hat{f},\lambda\}$ and $W$. Thus at a high-level, a solution algorithm must consist of alternating between two convex sub-problems, defined by the objective/decision variable pairs $(J_d, \{\hat{f},\hat{B},\lambda\})$ and $(J_m, \{W,\wl,\wu\})$.

\vspace{-3mm}
\section{Solution Formulation}\label{sec:reg}
\vspace{-1mm}

When performing dynamics learning on a system that is a priori \emph{known} to be exponentially stabilizable at some strictly positive rate $\lambda$, the constrained problem formulation in~\eqref{prob_gen2} follows naturally given the assured \emph{existence} of a CCM. Albeit, the infinite-dimensional nature of the constraints is a considerable technical challenge, broadly falling under the class of \emph{semi-infinite} optimization~\cite{HettichKortanek1993}. Alternatively, for systems that are not globally exponentially stabilizable in $\X$, one can imagine that such a constrained formulation may lead to adverse effects on the learned dynamics model. 

Thus, in this section we propose a relaxation of problem~\eqref{prob_gen2} motivated by the concept of regularization. Specifically, constraints~\eqref{killing_A} and~\eqref{nat_contraction_W} capture this notion of stability of infinitesimal deviations \emph{at all points} in the space $\X$. They stem from requiring that $\dot{V} \leq -2\lambda V(x,\dx)$ in eq~\eqref{V_dot} when $\dx^T M(x) B(x) = 0$, i.e., when $\delta_u$ can have no effect on $\dot{V}$. This is nothing but the standard control Lyapunov inequality, applied to the differential setting. Constraint~\eqref{killing_A} sets to zero, the terms in~\eqref{V_dot} affine in $u$, while constraint~\eqref{nat_contraction_W} enforces this ``natural" stability condition. 

The simplifications we make are (i) relax constraints~\eqref{nat_contraction_W} and~\eqref{W_unif} to hold pointwise over some \emph{finite} constraint set $X_c \in \X$, and (ii) assume a specific sparsity structure for input matrix estimate $\hat{B}(x)$. We discuss the pointwise relaxation here; the sparsity assumption on $\hat{B}(x)$ is discussed in the following section and Appendix~\ref{app:justify_B}.

First, from a purely mathematical standpoint, the pointwise relaxation of~\eqref{nat_contraction_W} and \eqref{W_unif} is motivated by the observation that as the CCM-based controller is exponentially stabilizing, we only require the differential stability condition to hold locally (in a tube-like region) with respect to the provided demonstrations. By continuity of eigenvalues for continuously parameterized entries of a matrix, it is sufficient to enforce the matrix inequalities at a sampled set of points~\cite{Lax2007}.

Second, enforcing the existence of such an ``approximate" CCM seems to have an impressive regularization effect on the learned dynamics that is more meaningful than standard regularization techniques used in for instance, ridge or lasso regression. Specifically, problem~\eqref{prob_gen2}, and more generally, problem~\eqref{prob_gen} can be viewed as the \emph{projection} of the best-fit dynamics onto the set of stabilizable systems. This results in dynamics models that jointly balance regression performance and stabilizablity, ultimately yielding systems whose generated trajectories are notably easier to track. This effect of regularization is discussed in detail in our experiments in Section~\ref{sec:result}.

\revision{Practically, the finite constraint set $X_c$, with cardinality $N_c$, includes all $\xs$ in the regression training set of $\{(\xs,\us,\dot{x}_i)\}_{i=1}^{N}$ tuples. However, as the LMI constraints are \emph{independent} of $\us,\dot{x}_i$, the set $X_c$ is chosen as a strict superset of $\{\xs\}_{i=1}^{N}$ (i.e., $N_c > N$) by randomly sampling additional points within $\X$, drawing parallels with semi-supervised learning.}

\vspace{-2mm}
\subsection{Sparsity of Input Matrix Estimate $\hat{B}$} \label{sec:B_simp}
\vspace{-2mm}

While a pointwise relaxation for the matrix inequalities is reasonable, one cannot apply such a relaxation to the exact equality condition in~\eqref{killing_A}. Thus, the second simplification made is the following assumption, reminiscent of control normal form equations.
\begin{assumption}\label{ass:B_simp}
Assume $\hat{B}(x)$ to take the following sparse representation:
\begin{equation}
    \hat{B}(x) = \begin{bmatrix} O_{(n-m)\times m} \\ \bs(x) \end{bmatrix},
\label{B_simp}
\end{equation}
where $\bs(x)$ is an invertible $m\times m$ matrix for all $x\in \X$. 
\end{assumption}
For the assumed structure of $\hat{B}(x)$, a valid $B_{\perp}$ matrix is then given by:
\begin{equation}
    B_{\perp} = \begin{bmatrix} I_{n - m} \\ O_{m \times (n-m)} \end{bmatrix}.
    \label{B_perp}
\end{equation}
Therefore, constraint~\eqref{killing_A} simply becomes:
\[
	\partial_{\hat{b}_j} W_{\perp} (x) = 0, \quad j = 1,\ldots,m.
\]
where $W_{\perp}$ is the upper-left $(n-m)\times (n-m)$ block of $W(x)$. Assembling these constraints for the $(p,q)$ entry of $W_{\perp}$, i.e., $w_{\perp_{pq}}$, we obtain:
\[
	 \begin{bmatrix} \dfrac{ \partial w_{\perp_{pq}} (x) }{\partial x^{(n-m)+1}} & \cdots & \dfrac{\partial w_{\perp_{pq}} (x) }{\partial x^{n}} \end{bmatrix} \bs(x) = 0.
\]
Since the matrix $\bs(x)$ in~\eqref{B_simp} is assumed to be invertible, the \emph{only} solution to this equation is $\partial w_{\perp_{pq}}/ \partial x^i = 0$ for $i = (n-m)+1,\ldots,n$, and all $(p,q) \in \{1,\ldots,(n-m)\}$. That is, $W_{\perp}$ cannot be a function of the last $m$ components of $x$ -- an elegant simplification of constraint~\eqref{killing_A}. Due to space limitations, justification for this sparsity assumption is provided in Appendix~\ref{app:justify_B}.

\subsection{Finite-dimensional Optimization}\label{sec:finite}

We now present a tractable finite-dimensional optimization for solving problem~\eqref{prob_gen2} under the two simplifying assumptions \revision{introduced in the previous sections}. The derivation of the solution algorithm itself is presented in Appendix~\ref{sec:deriv}, and relies extensively on vector-valued Reproducing Kernel Hilbert Spaces. 

\begin{leftbox}
\begin{itemize}[leftmargin=0.4in]
    \item[{\bf Step 1:}] Parametrize the functions $\hat{f}$, the columns of $\hat{B}(x)$: $\{\hat{b}_j\}_{j=1}^{m}$, and $\{w_{ij}\}_{i,j=1}^{n}$ as a linear combination of features. That is, 
\begin{align}
    \hat{f}(x) &= \Phi_f(x)^T \alpha, \label{param_1}\\
    \hat{b}_j(x) &= \Phi_b(x)^T \beta_j  \quad j \in \{1,\ldots, m\}, \\
    w_{ij}(x) &= \begin{cases} \hat{\phi}_w(x)^T \hat{\theta}_{ij} &\text{ if }\quad  (i,j) \in \{1,\ldots,n-m\}, \\
    \phi_w(x)^T \theta_{ij} &\text{ else}, \label{param_2}
    \end{cases}
\end{align}
where $\alpha \in \reals^{d_f}$, $\beta_j \in \reals^{d_b}$, $\hat{\theta}_{ij}, \theta_{ij} \in \reals^{d_w}$ are constant vectors to be optimized over, and $\Phi_f : \X \rightarrow \reals^{d_f\times n}$, $\Phi_b : \X \rightarrow \reals^{d_b \times n}$, $\hat{\phi}_w : \X \rightarrow \reals^{d_w}$ and $\phi_w : \X \rightarrow \reals^{d_w}$ are a priori chosen feature mappings. To enforce the sparsity structure in~\eqref{B_simp}, the feature matrix $\Phi_b$ must have all 0s in its first $n-m$ columns. The features $\hat{\phi}_w$ are distinct from $\phi_w$ in that the former are only a function of the first $n-m$ components of $x$ (as per Section~\ref{sec:B_simp}).
While one can use any function approximator (e.g., neural nets), we motivate this parameterization from a perspective of Reproducing Kernel Hilbert Spaces (RKHS); see Appendix~\ref{sec:deriv}.
\newline
\item[{\bf Step 2:}] Given positive regularization constants $\mu_f, \mu_b, \mu_w$ and positive tolerances $(\delta_\lambda,\epsilon_\lambda)$ and $(\delta_{\wl}, \epsilon_{\wl})$, solve:
\begin{subequations}\label{learn_finite}
\begin{align}
    \min_{\alpha,\beta_j, \hat{\theta}_{ij}, \theta_{ij}, \wl, \wu,\lambda} \quad &  \overbrace{\sum_{k=1}^{N} \| \hat{f}(\xs)+\hat{B}(\xs)u_i - \dot{x}_i \|_2^2 + \mu_f \|\alpha\|_2^2 + \mu_b \sum_{j=1}^{m} \|\beta_j\|_2^2}^{:=J_d}  \nonumber \\
    \quad & \quad + \underbrace{(\wu-\wl) +  \mu_w\sum_{i,j} \|\tilde{\theta}_{ij}\|_2^2}_{:=J_m}  \\
    \text{s.t.} \quad & F(\xs;\alpha,\tilde{\theta}_{ij}, \lambda + \epsilon_{\lambda}) \preceq 0, \quad \forall \xs \in X_c, \label{nat_finite} \\
    \quad & (\wl + \epsilon_{\wl})I_{n} \preceq W(\xs) \preceq \wu I_n, \quad \forall \xs \in X_c, \label{uniform_finite} \\
    \quad & \theta_{ij} = \theta_{ji},  \hat{\theta}_{ij} = \hat{\theta}_{ji} \label{sym_finite} \\
    \quad &\lambda \geq \delta_{\lambda}, \quad  \wl \geq \delta_{\wl}, \label{tol_finite}
\end{align}
\end{subequations}
where $\tilde{\theta}_{ij}$ is used as a placeholder for $\theta_{ij}$ and $\hat{\theta}_{ij}$ to simplify notation.
\end{itemize}
\end{leftbox}

We wish to highlight the following key points regarding problem~\eqref{learn_finite}. 
Constraints \eqref{nat_finite} and~\eqref{uniform_finite} are the pointwise relaxations of~\eqref{nat_contraction_W} and~\eqref{W_unif} respectively. Constraint~\eqref{sym_finite} captures the fact that $W(x)$ is a symmetric matrix. Finally, constraint~\eqref{tol_finite} imposes some tolerance requirements to ensure a well conditioned solution. Additionally, the tolerances $\epsilon_{\delta}$ and $\epsilon_{\wl}$ are used to account for the pointwise relaxations of the matrix inequalities. A key challenge is to efficiently solve this constrained optimization problem, given a potentially large number of constraint points in $X_c$. In the next section, we present an iterative algorithm and an adaptive constraint sampling technique to solve problem~\eqref{learn_finite}.

%A class of underactuated systems captured by our dynamics representation cannot be stabilized around equilibrium points using time-invariant continuous state feedback. Thus, the pointwise relaxation is not only practical, but also necessary, since for such systems, one cannot hope to find a uniformly (i.e., even at equilibrium points) valid CCM.

%===============================================================================
\vspace{-2mm}
\section{Solution Algorithm} \label{sec:soln}
\vspace{-2mm}
The fundamental structure of the solution algorithm consists of alternating between the dynamics and metric sub-problems derived from problem~\eqref{learn_finite}. We also make a few additional modifications to aid tractability, most notable of which is the use of a \emph{dynamically} updating set of constraint points $X_c^{(k)}$ at which the LMI constraints are enforced at the $k^{\text{th}}$ iteration. In particular $X_c^{(k)} \subset X_c$ with $N_c^{(k)}:= |X_c^{(k)}|$ being ideally much less than $N_c$, the cardinality of the full constraint set $X_c$. Formally, each major iteration $k$ is characterized by three minor steps (sub-problems):
\begin{leftbox}
\begin{enumerate}
\item Finite-dimensional dynamics sub-problem at iteration $k$:
\begin{subequations} \label{finite_dyn}
\begin{align}
    \min_{\substack{\alpha,\beta_j, j=1,\ldots,m,\ \lambda \\ s \geq 0}} \quad & J_d(\alpha,\beta) + \mu_s\|s\|_1 \\
    \text{s.t.} \quad & F(\xs;\alpha,\tilde{\theta}^{(k-1)}_{ij}, \lambda + \epsilon_{\lambda}) \preceq s(\xs)I_{n-m}  \quad \forall \xs \in X_c^{(k)} \\
    \quad & s(\xs) \leq \bar{s}^{(k-1)}  \quad \forall \xs \in X_c^{(k)}\\
    \quad & \lambda \geq \delta_{\lambda},
\end{align}
\end{subequations}
where $\mu_s$ is an additional regularization parameter for $s$ -- an $N_c^{(k)}$ dimensional non-negative slack vector. The quantity $\bar{s}^{(k-1)}$ is defined as
\[
    \begin{split}
    \bar{s}^{(k-1)} &:= \max_{\xs \in X_c} \lambda_{\max} \left(F^{(k-1)}(\xs)\right), \quad \text{where} \\
    F^{(k-1)}(\xs) &:= F(\xs;\alpha^{(k-1)},\tilde{\theta}^{(k-1)}_{ij}, \lambda^{(k-1)} +\epsilon_{\lambda}).
    \end{split}
\]
That is, $\bar{s}^{(k-1)}$ captures the worst violation for the $F(\cdot)$ LMI over the entire constraint set $X_c$, given the parameters at the end of iteration $k-1$. 
\item Finite-dimensional metric sub-problem at iteration $k$:
\begin{subequations}\label{finite_met}
\begin{align}
    \min_{\tilde{\theta}_{ij},\wl,\wu,  s \geq 0} \quad & J_m(\tilde{\theta}_{ij},\wl,\wu) + (1/\mu_s)\|s\|_1 \\
    \text{s.t.} \quad & F(\xs;\alpha^{(k)},\tilde{\theta}_{ij}, \lambda^{(k)} + \epsilon_{\lambda}) \preceq s(\xs)I_{n-m}  \quad \forall \xs \in X_c^{(k)} \\
    \quad & s(\xs) \leq \bar{s}^{(k-1)} \quad \forall \xs \in X_c^{(k)} \\
    \quad &  (\wl + \epsilon_{\wl})I_{n} \preceq W(\xs) \preceq \wu I_n, \quad \forall \xs \in X_c^{(k)}, \\
    \quad & \wl \geq \delta_{\wl}.
\end{align}
\end{subequations}

\item Update $X_c^{(k)}$ sub-problem. Choose a tolerance parameter $\delta>0$. Then, define
    \[
        \nu^{(k)}(\xs) := \max \left\{ \lambda_{\max} \left(F^k(\xs)\right) , \lambda_{\max} \left((\delta_{\wl}+\epsilon_{\delta})I_n - W(\xs) \right) \right \}, \quad \forall \xs \in X_c,
    \]
    and set
    \begin{equation}
        X_{c}^{(k+1)} :=  \left\{ \xs \in X_c^{(k)} : \nu^{(k)}(\xs) > -\delta \right\} \bigcup  \left\{\xs \in X_c \setminus X_c^{(k)} : \nu^{(k)}(\xs) > 0 \right\}. 
        \label{Xc_up}
    \end{equation}
\end{enumerate}
\end{leftbox}
Thus, in the update $X_c^{(k)}$ step, we balance addressing points where constraints are being violated ($\nu^{(k)} > 0$) and discarding points where constraints are satisfied with sufficient strict inequality ($\nu^{(k)}\leq -\delta$). This prevents overfitting to any specific subset of the constraint points. A potential variation to the union above is to only add up to say $K$ constraint violating points from $X_c\setminus X_c^{(k)}$ (e.g., corresponding to the $K$ worst violators), where $K$ is a fixed positive integer. Indeed this is the variation used in our experiments and was found to be extremely efficient in balancing the size of the set $X_c^{(k)}$ and thus, the complexity of each iteration. This adaptive sampling technique is inspired by \emph{exchange algorithms} for semi-infinite optimization, as the one proposed in~\cite{ZhangWuEtAl2010} where one is trying to enforce the constraints at \emph{all} points in a compact set $\X$.

Note that after the first major iteration, we replace the regularization terms in $J_d$ and $J_m$ with $\|\alpha^{(k)} - \alpha^{(k-1)}\|_2^2$, $\|\beta_j^{(k)}-\beta_j^{(k-1)}\|_2^2$, and $\|\tilde{\theta}_{ij}^{(k)} - \tilde{\theta}_{ij}^{(k-1)}\|_2^2$. This is done to prevent large updates to the parameters, particularly due to the dynamically updating constraint set $X_c^{(k)}$. The full pseudocode is summarized below in Algorithm~\ref{alg:final}. 

\begin{algorithm}[h!]
  \caption{Stabilizable Non-Linear Dynamics Learning (SNDL)}
  \label{alg:final}
  \begin{algorithmic}[1]
  \State {\bf Input:} Dataset $\{\xs,\us,\dot{x}_i\}_{i=1}^{N}$, constraint set $X_c$, regularization constants $\{\mu_f,\mu_b,\mu_w\}$, constraint tolerances $\{\delta_\lambda,\epsilon_\lambda,\delta_{\wl},\epsilon_{\wl} \}$, discard tolerance parameter $\delta$, Initial \# of constraint points: $N_c^{(0)}$, Max \# iterations: $N_{\max}$, termination tolerance $\varepsilon$. 
   \State $k \leftarrow 0$, \texttt{converged} $\leftarrow$ \textbf{false}, $W(x) \leftarrow I_n$.
   \State $X_c^{(0)} \leftarrow \textproc{RandSample}(X_c,N_c^{(0)})$ \label{line:rand_samp_init}
   \While {$\neg \texttt{converged} \wedge k<N_{\max} $} 
    \State $\{\alpha^{(k)}, \beta_j^{(k)}, \lambda^{(k)} \} \leftarrow \textproc{Solve}$~\eqref{finite_dyn}
    \State $\{\tilde{\theta}_{ij}^{(k)},\wl,\wu\} \leftarrow \textproc{Solve}$~\eqref{finite_met}
    \State $X_c^{(k+1)}, \bar{s}^{(k)}, \nu^{(k)} \leftarrow$ \textproc{Update} $X_c^{(k)}$ using~\eqref{Xc_up}
    \State {\small $\Delta \leftarrow \max\left\{\|\alpha^{(k)}-\alpha^{(k-1)}\|_{\infty},\|\beta_j^{(k)}-\beta_j^{(k-1)}\|_{\infty},\|\tilde{\theta}_{ij}^{(k)}-\tilde{\theta}_{ij}^{(k-1)}\|_{\infty},\|\lambda^{(k)}-\lambda^{(k-1)}\|_{\infty} \right\}$}
    \If{$\Delta < \varepsilon$ \textbf{or} $\nu^{(k)}(\xs) < \varepsilon \quad \forall \xs \in X_c$}
        \State \texttt{converged} $\leftarrow$ \textbf{true}.
    \EndIf
    \State $k \leftarrow k + 1$.
  \EndWhile
      \end{algorithmic}
\end{algorithm} 

\revision{Some comments are in order. First, convergence in Algorithm~\ref{alg:final} is declared if either progress in the solution variables stalls or all constraints are satisfied within tolerance. Due to the semi-supervised nature of the algorithm in that the number of constraint points $N_c$ can be significantly larger than the number of supervisory regression tuples $N$, it is impractical to enforce constraints at all $N_c$ points in any one iteration. Two key consequences of this are: (i) the matrix function $W(x)$ at iteration $k$ resulting from variables $\tilde{\theta}^{(k)}$ does \emph{not} have to correspond to a valid dual CCM for the interim learned dynamics at iteration $k$, and (ii) convergence based on constraint satisfaction at all $N_c$ points is justified by the fact that at each iteration, we are solving relaxed sub-problems that collectively generate a sequence of lower-bounds on the overall objective. Potential future topics in this regard are: (i) investigate the properties of the converged dynamics for models that are a priori unknown unstabilizable, and (ii) derive sufficient conditions for convergence for both the infinitely- and finitely- constrained versions of problem~\eqref{prob_gen2}.

Second, as a consequence of this iterative procedure, the dual metric and contraction rate pair $\{W(x),\lambda\}$ do not possess any sort of ``control-theoretic'' optimality. For instance, in~\cite{SinghMajumdarEtAl2017}, for a known stabilizable dynamics model, both these quantities are optimized for robust control performance. In this work, these quantities are used solely as \emph{regularizers} to \emph{promote} stabilizability of the learned model. A potential future topic to explore in this regard is how to further optimize $\{W(x),\lambda\}$ for control \emph{performance} for the final learned dynamics.}

%===============================================================================
\vspace{-3mm}
\section{Experimental Results} \label{sec:result}
\vspace{-2mm}

In this section we validate our algorithms by benchmarking our results on a known dynamics model. Specifically, we consider the 6-state planar vertical-takeoff-vertical-landing (PVTOL) model. The system is defined by the state: $(p_x,p_z,\phi,v_x,v_z,\dot{\phi})$ where $(p_x,p_z)$ is the position in the 2D plane, $(v_x,v_z)$ is the body-reference velocity, $(\phi,\dot{\phi})$ are the roll and angular rate respectively, and 2-dimensional control input $u$ corresponding to the motor thrusts. The true dynamics are given by:
\[
    \dot{x}(t) = \begin{bmatrix} v_x \cos\phi - v_z \sin\phi \\ v_x\sin\phi + v_z\cos\phi \\ \dot{\phi} \\ v_z\dot{\phi} - g\sin\phi \\ -v_x\dot{\phi} - g\cos\phi \\ 0 \end{bmatrix} + \begin{bmatrix} 0&0\\0&0 \\0&0 \\0&0 \\ (1/m) &(1/m) \\ l/J & (-l/J) \end{bmatrix}u,
\]
where $g$ is the acceleration due to gravity, $m$ is the mass, $l$ is the moment-arm of the thrusters, and $J$ is the moment of inertia about the roll axis. 
%\begin{figure}[h]
%	\centering
%	\includegraphics[width=0.35\textwidth]{pvtol.png}
%	\caption{Definition of the PVTOL state variables and model parameters: $l$ denotes the thrust moment arm (symmetric).}
%	\label{fig:PVTOL}
%\end{figure} 
We note that typical benchmarks in this area of work either present results on the 2D LASA handwriting dataset~\cite{Khansari-ZadehBillard2011} or other low-dimensional motion primitive spaces, with the assumption of full robot dynamics invertibility. The planar quadrotor on the other hand is a complex non-minimum phase dynamical system that has been heavily featured within the acrobatic robotics literature and therefore serves as a suitable case-study. 

%Due to space constraints, we provide details of our implementation in the appendix. In summary, the training data was generated by fitting randomly sampled geometric paths with polynomial spline trajectories that were then tracked with a sub-optimal PD controller to emulate a noisy/imperfect demonstrator. We solved problem~\eqref{prob_gen2} using Algorithm~\ref{alg:final} and leveraging a matrix feature mapping derived from RKHS theory. The algorithm converged in 5 major iterations, and leveraged a constraint set size $N_c$ of at most 344 points for any of the major iterations. Compared with the $N=1814$ points in the training dataset, this was a substantial computational gain. 

\vspace{-2mm}
\subsection{Generation of Datasets} \label{sec:data_gen}

The training dataset was generated in 3 steps. First, a fixed set of waypoint paths in $(p_x,p_z)$ were randomly generated. Second, for each waypoint path, multiple smooth polynomial splines were fitted using a minimum-snap algorithm. To create variation amongst the splines, the waypoints were perturbed within Gaussian balls and the time durations for the polynomial segments were also randomly perturbed. Third, the PVTOL system was simulated with perturbed initial conditions and the polynomial trajectories as references, and tracked using a sub-optimally tuned PD controller; thereby emulating a noisy/imperfect demonstrator. These final simulated paths were sub-sampled at $0.1$s resolution to create the datasets. The variations created at each step of this process were sufficient to generate a rich exploration of the state-space for training.

Due to space constraints, we provide details of the solution parameterization (number
of features, etc) in Appendix~\ref{app:prob_params}.
\vspace{-2mm}
\subsection{Models}
Using the same feature space, we trained three separate models with varying training dataset (i.e., $(\xs,\us,\dot{x}_s)$ tuples) sizes of $N \in \{100, 250, 500, 1000\}$. \revision{The first model, {\bf N-R} was an unconstrained and un-regularized model, trained by solving problem~\eqref{finite_dyn} without constraints or $l_2$ regularization (i.e., just least-squares).} The second model, {\bf R-R} was an unconstrained ridge-regression model, trained by solving problem~\eqref{finite_dyn} without any constraints (i.e., least-squares plus $l_2$ regularization). The third model, {\bf CCM-R} is the CCM-regularized model, trained using Algorithm~\ref{alg:final}. \revision{We enforced the CCM regularizing constraints for the CCM-R model at $N_c = 2400$ points in the state-space, composed of the $N$ demonstration points in the training dataset and randomly sampled points from $\X$ (recall that the CCM constraints do not require samples of $u,\dot{x}$). }

\revision{As the CCM constraints were relaxed to hold pointwise on the finite constraint set $X_c$ as opposed to everywhere on $\X$, in the spirit of viewing these constraints as regularizers for the model (see Section~\ref{sec:reg}), we simulated both the R-R and CCM-R models using the time-varying Linear-Quadratic-Regulator (TV-LQR) feedback controller.} This also helped ensure a more direct comparison of the quality of the learned models themselves, independently of the tracking feedback controller. \revision{The results are virtually identical using a tracking MPC controller and yield no additional insight.}
\vspace{-2mm}
\subsection{Validation and Comparison}\label{sec:verify}

The validation tests were conducted by gridding the $(p_x,p_z)$ plane to create a set of 120 initial conditions between 4m and 12m away from $(0,0)$ and randomly sampling the other states for the rest of the initial conditions. These conditions were \emph{held fixed} for both models and for all training dataset sizes to evaluate model improvement.

\revision{For each model at each value of $N$}, the evaluation task was to (i) solve a trajectory optimization problem to compute a dynamically feasible trajectory for the learned model to go from initial state $x_0$ to the goal state - a stable hover at $(0,0)$ at near-zero velocity; and (ii) track this trajectory using the TV-LQR controller. As a baseline, all simulations without \revision{any feedback controller (i.e., open-loop control rollouts) led to the PVTOL crashing}. This is understandable since the dynamics fitting objective is not optimizing for \emph{multi-step} error. \revision{The trajectory optimization step was solved as a fixed-endpoint, fixed final time optimal control problem using the Chebyshev pseudospectral method~\cite{FahrooRoss2002} with the objective of minimizing $\int_{0}^T \|u(t)\|^2 dt$. The final time $T$ for a given initial condition was held fixed between all models. Note that 120 trajectory optimization problems were solved for each model and each value of $N$.}

Figure~\ref{fig:box_all} shows a boxplot comparison of the trajectory-wise RMS full state errors ($\|x(t)-x^*(t)\|_2$ where $x^*(t)$ is the reference trajectory obtained from the optimizer and $x(t)$ is the actual realized trajectory) for each model and all training dataset sizes. 
\begin{figure}[h]
    \centering
    \includegraphics[width=\textwidth,clip]{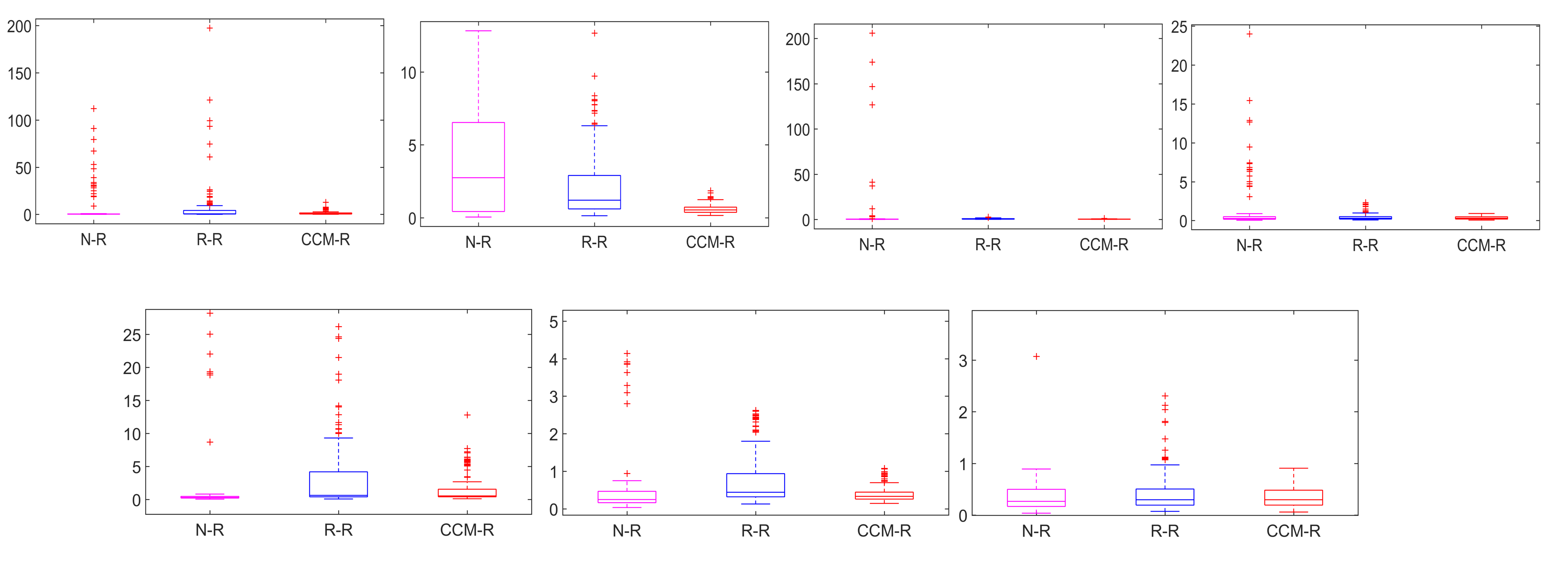}
    \caption{Box-whisker plot comparison of trajectory-wise RMS state-tracking errors over all 120 trajectories for each model and all training dataset sizes. \emph{Top row, left-to-right:} $N=100, 250, 500, 1000$; \emph{Bottom row, left-to-right:} $N=100, 500, 1000$ (zoomed in). The box edges correspond to the $25$th, median, and $75$th percentiles; the whiskers extend beyond the box for an additional 1.5 times the interquartile range; outliers, classified as trajectories with RMS errors past this range, are marked with red crosses. Notice the presence of unstable trajectories for N-R at all values of $N$ and for R-R at $N=100, 250$. The CCM-R model dominates the other two \emph{at all values of $N$}, particularly for $N = 100, 250$. }
        \label{fig:box_all}
\end{figure}
\revision{
As $N$ increases, the spread of the RMS errors decreases for both R-R and CCM-R models as expected. However, we see that the N-R model generates \emph{several} unstable trajectories for $N=100, 500$ and $1000$, indicating the need for \emph{some} form of regularization. The CCM-R model consistently achieves a lower RMS error distribution than both the N-R and R-R models \emph{for all training dataset sizes}. Most notable however, is its performance when the number of training samples is small (i.e., $N \in \{100, 250\}$) when there is considerable risk of overfitting. It appears the CCM constraints have a notable effect on the \emph{stabilizability} of the resulting model trajectories (recall that the initial conditions of the trajectories and the tracking controllers are held fixed between the models). 

For $N=100$ (which is really at the extreme lower limit of necessary number of samples since there are effectively $97$ features for each dimension of the dynamics function), both N-R and R-R models generate a large number of unstable trajectories. In contrast, out of the 120 generated test trajectories, the CCM-R model generates \emph{one} mildly (in that the quadrotor diverged from the nominal trajectory but did not crash) unstable trajectory. No instabilities were observed with CCM-R for $N \in \{250, 500, 1000\}$. 

Figure~\ref{fig:traj_100_uncon} compares the $(p_x,p_z)$ traces between R-R and CCM-R corresponding to the five worst performing trajectories for the R-R $N=100$ model. Similarly, Figure~\ref{fig:traj_100_CCM} compares the $(p_x,p_z)$ traces corresponding to the five worst performing trajectories for the CCM-R $N=100$ model. Notice the large number of unstable trajectories generated using the R-R model. Indeed, it is in this low sample training regime where the control-theoretic regularization effects of the CCM-R model are most noticeable. 
%The $(p_x,p_z)$ trajectory comparisons for $N=250$ are presented in Figure~\ref{fig:traj_250}. Specifically, on the left, we show the nominal (dashed) reference trajectories versus the actual realized (solid) trajectories for a subset of the initial conditions for the $N=250$ R-R model. The figure on the right shows the corresponding plot for the CCM-R model. Notice how not only is the tracking poor for the R-R model, but the nominal trajectories generated by the optimizer are quite jagged and unnatural for the true vehicle. In comparison, the CCM-R model does a significantly better job at both tasks, \emph{despite the low number of training samples.}
}
\begin{figure}[h]
	\centering
	\begin{subfigure}[t]{0.8\textwidth}
		\centering
		\includegraphics[width=\textwidth,clip]{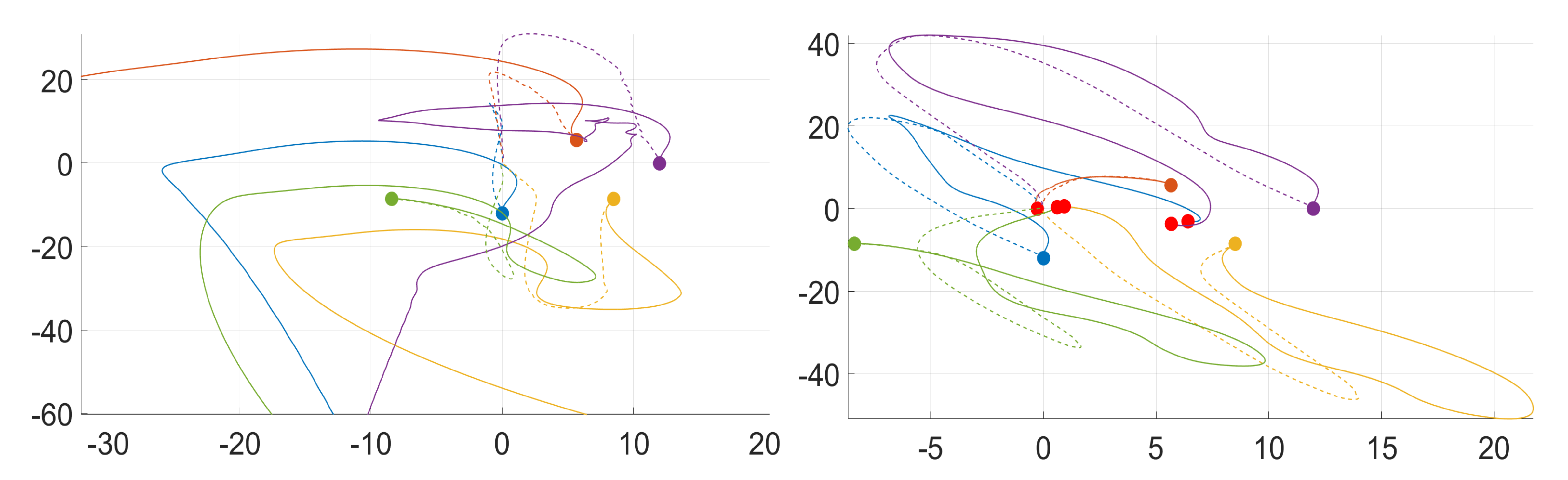}
		\caption{}
		\label{fig:traj_100_uncon}
	\end{subfigure} \qquad
	\begin{subfigure}[t]{0.8\textwidth}
		\centering
		\includegraphics[width=\textwidth,clip]{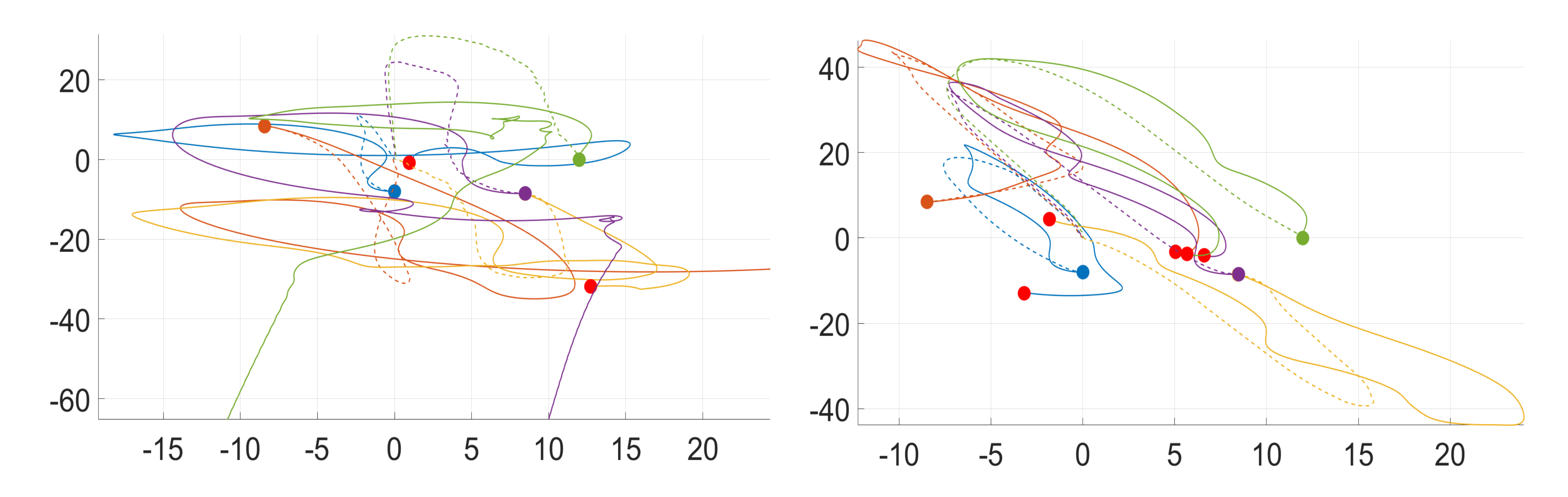}
		\caption{}
		\label{fig:traj_100_CCM}
	\end{subfigure}	
    \caption{ $(p_x,p_z)$ traces for R-R (\emph{left column}) and CCM-R (\emph{right column}) corresponding to the 5 worst performing trajectories for (a) R-R, and (b) CCM-R models at $N=100$. Colored circles indicate start of trajectory. Red circles indicate end of trajectory. All except one of the R-R trajectories are unstable. One trajectory for CCM-R is slightly unstable.}
        \label{fig:traj_250}
\end{figure}

Finally, in Figure~\ref{fig:unstable}, we highlight two trajectories, starting from the \emph{same initial conditions}, one generated and tracked using the R-R model, the other using the CCM model, for \revision{$N=250$}. Overlaid on the plot are the snapshots of the vehicle outline itself, illustrating the quite aggressive flight-regime of the trajectories \revision{(the initial starting bank angle is $40^\mathrm{o}$)}. While tracking the R-R model generated trajectory eventually ends in \revision{complete loss of control}, the system successfully tracks the CCM-R model generated trajectory to the stable hover at $(0,0$).

\begin{figure}[h]
    \centering
    \includegraphics[width=0.9\textwidth,clip]{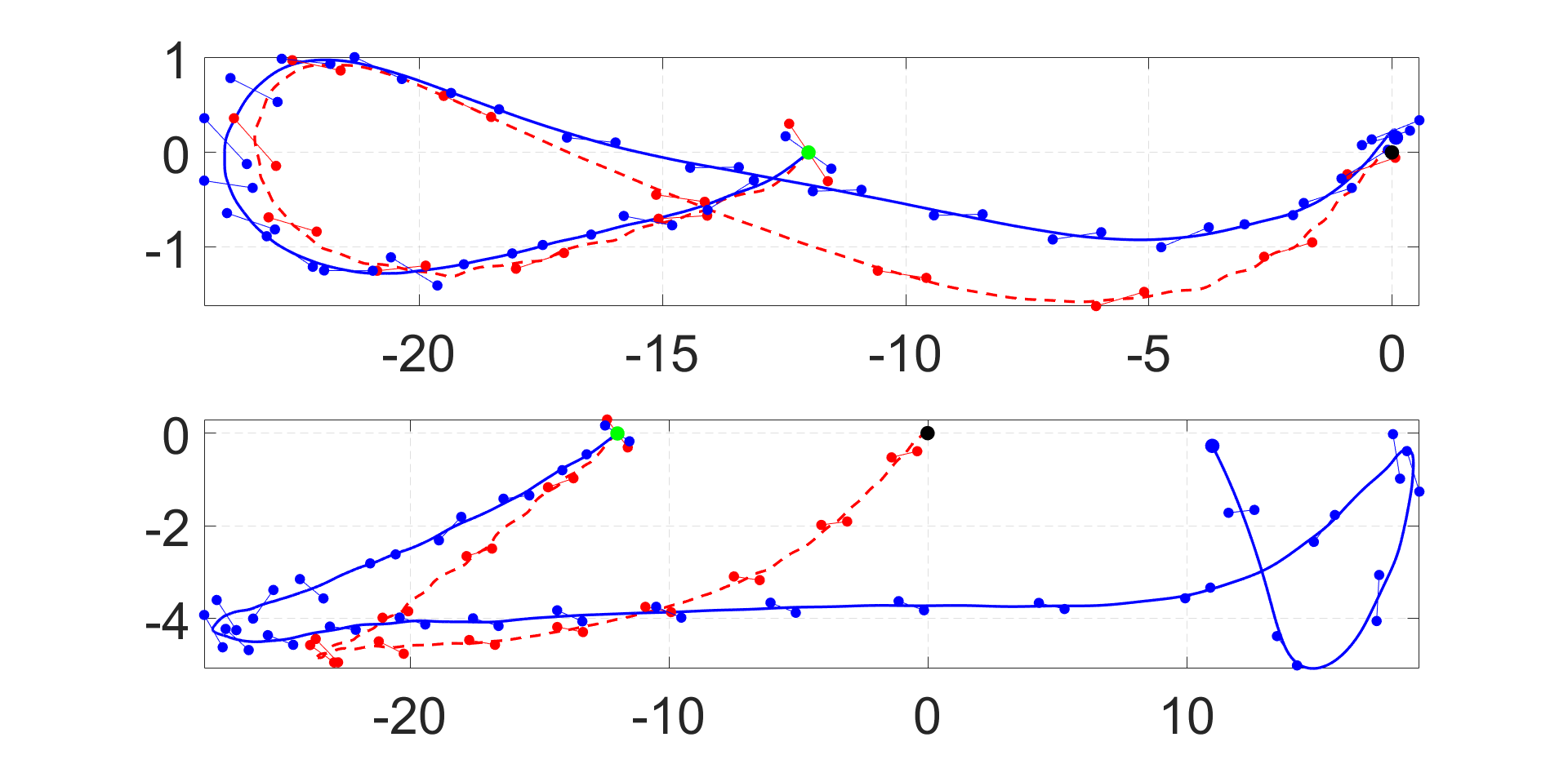}
    \caption{Comparison of reference and tracked trajectories in the $(p_x,p_z)$ plane for R-R and CCM-R models starting at same initial conditions with $N=250$. Red (dashed): nominal, Blue (solid): actual, Green dot: start, black dot: nominal endpoint, blue dot: actual endpoint; \emph{Top:} CCM-R, \emph{Bottom:} R-R. The vehicle successfully tracks the CCM-R model generated trajectory to the stable hover at $(0,0)$ while losing control when attempting to track the R-R model generated trajectory.}
        \label{fig:unstable}
\end{figure}

\revision{
An interesting area of future work here is to investigate how to tune the regularization parameters $\mu_f, \mu_b, \mu_w$. Indeed, the R-R model appears to be extremely sensitive to $\mu_f$, yielding drastically worse results with a small change in this parameter. On the other hand, the CCM-R model appears to be quite robust to variations in this parameter. Standard cross-validation techniques using regression quality as a metric are unsuitable as a tuning technique here; indeed, recent results even advocate for ``ridgeless'' regression~\cite{LiangRakhlin2018}. However, as observed in Figure~\ref{fig:box_all}, un-regularized model fitting is clearly unsuitable. The effect of regularization on how the trajectory optimizer leverages the learned dynamics is a non-trivial relationship that merits further study.}

\section{Conclusions}
In this paper, we presented a framework for learning \emph{controlled} dynamics from demonstrations for the purpose of trajectory optimization and control for continuous robotic tasks. By leveraging tools from nonlinear control theory, chiefly, contraction theory, we introduced the concept of learning \emph{stabilizable} dynamics, a notion which guarantees the existence of feedback controllers for the learned dynamics model that ensures trajectory trackability. 
Borrowing tools from  Reproducing Kernel Hilbert Spaces and convex optimization, we proposed a bi-convex semi-supervised algorithm for learning stabilizable dynamics for complex underactuated and inherently unstable systems. The algorithm was validated on a simulated planar quadrotor system where it was observed that our control-theoretic dynamics learning algorithm notably outperformed traditional ridge-regression based model learning.

There are several interesting avenues for future work. First, it is unclear how the algorithm would perform for systems that are fundamentally unstabilizable and how the resulting learned dynamics could be used for ``approximate'' control. Second, we will explore sufficient conditions for convergence for the iterative algorithm under the finite- and infinite-constrained formulations. Third, we will address extending the algorithm to work on higher-dimensional spaces through functional parameterization of the control-theoretic regularizing constraints. Fourth, we will address the limitations imposed by the sparsity assumption on the input matrix $B$ using the proposed alternating algorithm proposed in Section~\ref{sec:B_simp}. Finally, we will incorporate data gathered on a physical system subject to noise and other difficult to capture nonlinear effects (e.g., drag, friction, backlash) and validate the resulting dynamics model and tracking controllers on the system itself to evaluate the robustness of the learned models.

% The acknowledgments are autatically included only in the final version of the paper.
%\acknowledgments{If a paper is accepted, the final camera-ready version will (and probably should) include acknowledgments. All acknowledgments go at the end of the paper, including thanks to reviewers who gave useful comments, to colleagues who contributed to the ideas, and to funding agencies and corporate sponsors that provided financial support.}

%===============================================================================
\vspace{-3mm}
\renewcommand{\baselinestretch}{0.85}
\bibliographystyle{splncs03}
%\bibliography{../../../bib/main,../../../bib/ASL_papers} 

\begin{thebibliography}{10}
\providecommand{\url}[1]{\texttt{#1}}
\providecommand{\urlprefix}{URL }

\bibitem{BerkenkampTurchettaEtAl2017}
Berkenkamp, F., Turchetta, M., Schoellig, A., Krause, A.: Safe model-based
  reinforcement learning with stability guarantees. In: {Conf.\ on Neural
  Information Processing Systems} (2017)

\bibitem{BraultHeinonenEtAl2016}
Brault, R., Heinonen, M., d'Alch{\'e} Buc, F.: Random fourier features for
  operator-valued kernels. In: {Asian Conf. on Machine Learning}. pp. 110--125
  (2016)

\bibitem{ChuaCalandraEtAl2018}
Chua, K., Calandra, R., McAllister, R., Levine, S.: Deep reinforcement learning
  in a handful of trials using probabilistic dynamics models. arXiv preprint
  arXiv:1805.12114  (2018)

\bibitem{CrouchSchaft1987}
Crouch, P.E., van~der Schaft, A.J.: Variational and Hamiltonian Control
  Systems. {Springer} (1987)

\bibitem{DeisenrothRasmussen2011}
Deisenroth, M.P., Rasmussen, C.E.: {PILCO}: A model-based and data-efficient
  approach to policy search. In: {Int.\ Conf.\ on Machine Learning}. pp.
  465--472 (2011)

\bibitem{FahrooRoss2002}
Fahroo, F., Ross, I.M.: Direct trajectory optimization by a chebyshev
  pseudospectral method. {AIAA Journal of Guidance, Control, and Dynamics}
  25(1),  160--166 (2002)

\bibitem{HearstDumaisEtAl1998}
Hearst, M.A., Dumais, S.T., Osuna, E., Platt, J., Scholkopf, B.: Support vector
  machines. {IEEE Intelligent Systems and their Applications}  13(4),  18--28
  (1998)

\bibitem{HettichKortanek1993}
Hettich, R., Kortanek, K.O.: Semi-infinite programming: Theory, methods, and
  applications. {SIAM Review}  35(3),  380--429 (1993)

\bibitem{HuangAvronEtAl2014}
Huang, P.S., Avron, H., Sainath, T.N., Sindhwani, V., Ramabhadran, B.: Kernel
  methods match deep neural networks on {TIMIT}. In: {IEEE Int.\ Conf.\ on
  Acoustics, Speech and Signal Processing}. pp. 205--209. IEEE (2014)

\bibitem{Khansari-ZadehBillard2011}
Khansari-Zadeh, S.M., Billard, A.: Learning stable nonlinear dynamical systems
  with {Gaussian} mixture models. {IEEE Transactions on Robotics}  27(5),
  943--957 (2011)

\bibitem{Khansari-ZadehKhatib2017}
Khansari-Zadeh, S.M., Khatib, O.: Learning potential functions from human
  demonstrations with encapsulated dynamic and compliant behaviors. {Autonomous
  Robots}  41(1),  45--69 (2017)

\bibitem{KurdilaZabarankin2006}
Kurdila, A.J., Zabarankin, M.: Convex functional analysis. Springer Science \&
  Business Media, {Birkh\"{a}user Basel} edn. (2006)

\bibitem{Lax2007}
Lax, P.: Linear Algebra and its Applications. {John Wiley \& Sons}, 2 edn.
  (2007)

\bibitem{LemmeNeumannEtAl2014}
Lemme, A., Neumann, K., Reinhart, R.F., Steil, J.J.: Neural learning of vector
  fields for encoding stable dynamical systems. {Neurocomputing}  141(1),
  3--14 (2014)

\bibitem{LiangRakhlin2018}
Liang, T., Rakhlin, A.: Just interpolate: Kernel ridgeless regression can
  generalize. arXiv preprint arXiv:1808.00387v1  (2018)

\bibitem{LohmillerSlotine1998}
Lohmiller, W., Slotine, J.J.E.: On contraction analysis for non-linear systems.
  {Automatica}  34(6),  683--696 (1998)

\bibitem{ManchesterSlotine2017}
Manchester, I.R., Slotine, J.J.E.: Control contraction metrics: Convex and
  intrinsic criteria for nonlinear feedback design. {IEEE Transactions on
  Automatic Control}  (2017), {In Press}

\bibitem{ManchesterSlotine2018}
Manchester, I.R., Slotine, J.J.E.: Robust control contraction metrics: A convex
  approach to nonlinear state-feedback {$H-\infty$} control. {IEEE Control
  Systems Letters}  2(2),  333--338 (2018)

\bibitem{ManchesterTangEtAl2015}
Manchester, I., Tang, J.Z., Slotine, J.J.E.: Unifying classical and
  optimization-based methods for robot tracking control with control
  contraction metrics. In: {Int.\ Symp.\ on Robotics Research} (2015)

\bibitem{MedinaBillard2017}
Medina, J.R., Billard, A.: Learning stable task sequences from demonstration
  with linear parameter varying systems and hidden {Markov} models. In: {Conf.\
  on Robot Learning}. pp. 175--184 (2017)

\bibitem{MicheliGlaunes2014}
Micheli, M., Glaun{\'e}s, J.A.: Matrix-valued kernels for shape deformation
  analysis. {Geometry, Imaging and Computing}  1(1),  57--139 (2014)

\bibitem{Minh2016}
Minh, H.Q.: Operator-valued bochner theorem, fourier feature maps for
  operator-valued kernels, and vector-valued learning. arXiv preprint
  arXiv:1608.05639  (2016)

\bibitem{NagabandiKahnEtAl2017}
Nagabandi, A., Kahn, G., Fearing, R.S., Levine, S.: Neural network dynamics for
  model-based deep reinforcement learning with model-free fine-tuning. arXiv
  preprint arXiv:1708.02596  (2017)

\bibitem{Olfati-Saber2001}
Olfati-Saber, R.: Nonlinear Control of Underactuated Mechanical Systems with
  Application to Robotics and Aerospace Vehicles. Ph.D. thesis, {Massachusetts
  Inst.\ of Technology} (2001)

\bibitem{RahimiRecht2007}
Rahimi, A., Recht, B.: Random features for large-scale kernel machines. In:
  {Conf.\ on Neural Information Processing Systems}. pp. 1177--1184 (2007)

\bibitem{RahimiRecht2008}
Rahimi, A., Recht, B.: Uniform approximation of functions with random bases.
  In: {Allerton Conf.\ on Communications, Control and Computing}. IEEE (2008)

\bibitem{RavichandarSalehiEtAl2017}
Ravichandar, H., Salehi, I., Dani, A.: Learning partially contracting dynamical
  systems from demonstrations. In: {Conf.\ on Robot Learning} (2017)

\bibitem{ReyhanogluSchaftEtAl1999}
Reyhanoglu, M., van~der Schaft, A., McClamroch, N.H., Kolmanovsky, I.: Dynamics
  and control of a class of underactuated mechanical systems. {IEEE
  Transactions on Automatic Control}  44(9),  1663--1671 (1999)

\bibitem{SannerSlotine1992}
Sanner, R.M., Slotine, J.J.E.: Gaussian networks for direct adaptive control.
  {IEEE Transactions on Neural Networks}  3(6),  837--863 (1992)

\bibitem{ScholkoepfSmola2001}
Scholk{\"o}pf, B., Smola, A.J.: Learning with kernels: support vector machines,
  regularization, optimization, and beyond. {MIT Press} (2001)

\bibitem{SindhwaniTuEtAl2018}
Sindhwani, V., Tu, S., Khansari, M.: Learning contracting vector fields for
  stable imitation learning. arXiv preprint arXiv:1804.04878  (2018)

\bibitem{SinghMajumdarEtAl2017}
Singh, S., Majumdar, A., Slotine, J.J.E., Pavone, M.: Robust online motion
  planning via contraction theory and convex optimization. In: {Proc.\ IEEE
  Conf.\ on Robotics and Automation} (2017), {Extended Version, Available at
  }\url{http://asl.stanford.edu/wp-content/papercite-data/pdf/Singh.Majumdar.Slotine.Pavone.ICRA17.pdf}

\bibitem{SlotineLi1987}
Slotine, J.J.E., Li, W.: On the adaptive control of robot manipulators. {Int.\
  Journal of Robotics Research}  6(3),  49--59 (1987)

\bibitem{Spong1998}
Spong, M.W.: Underactuated mechanical systems. In: Control Problems in Robotics
  and Automation. {Springer Berlin Heidelberg} (1998)

\bibitem{VenkatramanCapobiancoEtAl2016}
Venkatraman, A., Capobianco, R., Pinto, L., Hebert, M., Nardi, D., Bagnell, A.:
  Improved learning of dynamics models for control. In: {Int.\ Symp.\ on
  Experimental Robotics}. pp. 703--713. Springer (2016)

\bibitem{VenkatramanHebertEtAl2015}
Venkatraman, A., Hebert, M., Bagnell, J.A.: Improving multi-step prediction of
  learned time series models. In: {Proc.\ AAAI Conf.\ on Artificial
  Intelligence} (2015)

\bibitem{ZhangWuEtAl2010}
Zhang, L., Wu, S.Y., L{\'o}pez, M.A.: A new exchange method for convex
  semi-infinite programming. {SIAM Journal on Optimization}  20(6),  2959--2977
  (2010)

\bibitem{Zhou2008}
Zhou, D.X.: Derivative reproducing properties for kernel methods in learning
  theory. {Journal of Computational and Applied Mathematics}  220(1-2),
  456--463 (2008)

\end{thebibliography}
\newcommand{\noopsort}[1]{} \newcommand{\printfirst}[2]{#1}
  \newcommand{\singleletter}[1]{#1} \newcommand{\switchargs}[2]{#2#1}

\newpage
\appendix
\renewcommand{\baselinestretch}{0.91}
\section*{Appendix}
\section{Justification for Sparsity Assumption for $\hat{B}$}
\label{app:justify_B}

Physically, structural assumption~\eqref{B_simp} is not as mysterious as it appears. Indeed, consider the standard dynamics form for mechanical systems:
\[
	H(q) \ddot{q} + C(q,\dot{q}) \dot{q} + g(q) = B(q) u,
\]
where $q \in \reals^{n_q}$ is the configuration vector, $H \in \Sjpp_{n_q}$ is the inertia matrix, $C(q,\dot{q})\dot{q}$ contains the centrifugal and Coriolis terms, $g$ are the gravitational terms, $B \in \reals^{n_q \times m}$ is the (full rank) input matrix mapping and $u \in \reals^{m}$ is the input. For fully actuated systems, $\mathrm{rank}(B) = m = n_q$. For underactuated systems, $m < n_q$. By rearranging the configuration vector~\cite{Spong1998,Olfati-Saber2001,ReyhanogluSchaftEtAl1999}, one can partition $q$ as $(q_u, q_a)$ where $q_u \in \reals^{n_q - m}$ represents the unactuated degrees of freedom and $q_a \in \reals^{m}$ represents the actuated degrees of freedom. Applying this partitioning to the dynamics equation above yields
\[
	\begin{bmatrix} H_{uu} (q) & H_{ua}(q) \\ H_{ua}(q) & H_{aa} (q)\end{bmatrix} \begin{bmatrix} \ddot{q}_u \\ \ddot{q}_a \end{bmatrix} + \begin{bmatrix} \tau_u(q,\dot{q}) \\ \tau_a (q,\dot{q}) \end{bmatrix} = \begin{bmatrix} O_{(n_q-m)\times m} \\ b(q) \end{bmatrix} u
\]
where $b \in \reals^{m \times m}$ is an invertible square matrix. As observed in~\cite{ReyhanogluSchaftEtAl1999}, a substantial class of underactuated systems can be represented in this manner. Of course, fully actuated systems also take this form ($m = n_q$). Thus, by taking as state $x = (q, p) \in \reals^{n}$ where $p = H(q) \dot{q}$ is momentum (so that $n = 2n_q$), the dynamics can be written as~\eqref{dyn}:
\begin{equation}
	\dot{x} = \begin{bmatrix} \dot{q} \\ \dot{p} \end{bmatrix} = \begin{bmatrix} H^{-1} (q) p \\ \dot{H}(q) \dot{q} - \tau(q,\dot{q}) \end{bmatrix}  + \begin{bmatrix} O_{(n-m) \times m} \\ b(q) \end{bmatrix}u.
\label{under_a}
\end{equation}
Notice that the input matrix takes the desired normal form in~\eqref{B_simp}. To address the apparent difficult of working with the state representation of $(q,p)$ (when usually only measurements of $(q,\dot{q},\ddot{q})$ are typically available \emph{and} the inertia matrix $H(q)$ is unknown), we make use of the following result from~\cite{ManchesterSlotine2017}:

\begin{theorem}[CCM Invariance to Diffeomorphisms]\label{thm:ccm_inv}
Suppose there exists a valid CCM $M_x(x)$ with respect to the state $x$. Then, if $z = \psi(x)$ is a diffeomorphism, then, the CCM conditions also hold with respect to state $z$ with metric $M_z(z) = \Psi(z)^{-T}M_x(x)\Psi(z)^{-1}$, where $\Psi(z) = \partial \psi(x)/\partial x$ evaluated at $x = \psi^{-1}(z)$.
\end{theorem}

Thus, for the (substantial) class of underactuated systems of the form~\eqref{under_a}, one would solve problem~\eqref{prob_gen2} by alternating between fitting $H, \tau, b$ (using the demonstrations $(q,\dot{q},\ddot{q})$) and leveraging Theorem~\ref{thm:ccm_inv} and the previous estimate of $H$ to enforce the matrix inequality constraints using the state representation $(q,p)$. This allows us to borrow several existing results from adaptive control on estimating mechanical system by leveraging the known linearity of $H(q)$ in terms of unknown mass property parameters multiplying known physical basis functions. We leave this extension however, to future work.

\section{Derivation of Problem~\eqref{learn_finite}}\label{sec:deriv}

To go from the general problem definition in~\eqref{prob_gen2} to the finite dimensional problem in~\eqref{learn_finite}, we first must define appropriate function classes for $f$, $b_j$, and $W$. We will do this using the framework of RKHS.

\subsection{Reproducing Kernel Hilbert Spaces}

{\bf Scalar-valued RKHS}: Kernel methods~~\cite{ScholkoepfSmola2001} constitute a broad family of non-parametric modeling techniques for solving a range of problems in machine learning. A scalar-valued positive definite kernel function $k : \X \times \X \mapsto \reals$ generates a Reproducing Kernel Hilbert Space (RKHS) of functions, with the nice property that if two functions are close in the distance derived from the norm (associated with the Hilbert space), then their pointwise evaluations are close at all points. This continuity of evaluation functionals has the far reaching consequence that norm-regularized learning problems over RKHSs admit finite dimensional solutions via Representer Theorems. The kernel $k$ may be (non-uniquely) associated with a higher-dimensional embedding of the input space via a feature map, $\phi: \reals^n \mapsto \reals^D$, such that $k(x, z) = \langle \phi(x), \phi(z)\rangle$, where $D$ is infinite for universal kernels associated with RKHSs that are dense in the space of square integrable functions. Standard regularized linear statistical models in the embedding space implicitly provide non-linear inference with respect to the original input representation.  In a nutshell, kernel methods provide a rigorous algorithmic framework for deriving non-linear counterparts of a whole array of linear statistical techniques, e.g. for classification, regression, dimensionality reduction and unsupervised learning. For details, we point the reader to~\cite{HearstDumaisEtAl1998}.

%\ssmargin{[insert paragraph]}{sumeet}

{\bf Vector-valued RKHS}: Dynamics estimation is a vector-valued learning problem. Such problems can be naturally formulated in terms of  vector-valued generalizations of RKHS conceps. The theory and formalism of vector-valued RKHS can be traced as far back as the work of Laurent Schwarz in 1964, with applications ranging from solving partial differential equations to machine learning. Informally, we say that that ${\cal H}$ is an RKHS of $\reals^n$-valued maps if for any $v\in \reals^n$, the linear functional that maps $f \in {\cal H}$ to $v^T f(x)$ is continuous.

More formally, denote the standard inner product on $\Y$ as $\ip{\cdot}{\cdot}$ and let $\Y(\X)$ be the vector space of all functions $f : \X \rightarrow \Y$ and let $\Lin(\Y)$ be the space of all bounded linear operators on $\Y$, i.e., $n\times n$ matrices. A function $K : \X \times \X \rightarrow \Lin(\Y)$ is an \emph{operator-valued positive definite kernel} if for all $(x,z) \in \X \times \X$, $K(x,z)^T = K(z,x)$ and for all finite set of points $\{x_i\}_{i=1}^{N} \in \X$ and $\{y_i\}_{i=1}^{N} \in \Y$, 
\[
	\sum_{i,j=1}^{N} \ip{y_i}{K(x_i,x_j)y_j} \geq 0.
\]
Given such a $K$, we define the unique $\Y$-valued RKHS $\Hk \subset \Y(\X)$ with reproducing kernel $K$ as follows. For each $x \in \X$ and $y\in \Y$, define the function $K_{x} y = K(\cdot, x) y \in \Y(\X)$. That is, for all $z \in \X$, $K_x y (z) = K(z,x) y$. Then, the Hilbert space $\Hk$ is defined to be the completion of the linear span of functions $\{K_x y \ | \ x \in \X, y \in \Y\}$ with inner product between functions $f  = \sum_{i=1}^{N} K_{x_i} y_i,\ g = \sum_{j=1}^{M} K_{z_j} w_i \in \Hk$, defined as:
\[
	\ip{f}{g}_{\Hk} = \sum_{i=1}^{N}\sum_{j=1}^{M} \ip{y_i}{K(x_i,z_j) w_j}.
\] 
Notably, the kernel $K$ satisfies the following reproducing property:
\begin{equation}
	\ip{f(x)}{y} = \ip{f}{K_x y}_{\Hk} \quad \forall (x,y) \in \X \times \Y \wedge f \in \Hk,
\label{RK_rep}
\end{equation}
and is thereby referred to as the \emph{reproducing kernel} for $\Hk$. As our learning problems involves the Jacobian of $f$, we will also require a representation for the derivative of the kernel matrix. Accordingly, suppose the kernel matrix $K$ lies in $C^{2}(\X \times \X)$. For any $j \in \{1,\ldots, n\}$, define the matrix functions $\frac{\partial}{\partial s^j} K: \X \times \X \rightarrow \Lin(\Y)$ and $\frac{\partial}{\partial r^j} K : \X \times \X \rightarrow \Lin(\Y)$ as
\[
	\dfrac{\partial}{\partial s^j} K(z,x) := \left.\dfrac{\partial}{\partial s^j} K(r , s) \right|_{r = z, s = x}  , \quad \dfrac{\partial}{\partial r^j} K(z,x) := \left. \dfrac{\partial}{\partial r^j} K(r,s) \right|_{r = z,s=x},
\]
where the derivative of the kernel matrix is understood to be element-wise. 
%By the adjoint property of $K$, we have the identity
% \[
% 	\dfrac{\partial}{\partial r^j} K(z,x) = \dfrac{\partial}{\partial s^j} K^T(x,z).
% \]
Define the $C^1$ function $\partial_j K_x y$ on $\Y(\X)$ as
\[
	\partial_j K_x y(z) := \dfrac{\partial}{\partial s^j} K(z, x)y \quad \forall z \in \X.
\]
% Similarly, for any $i,j \in \{1,\ldots,n\}$ define the matrix function $\frac{\partial^2}{\partial r^i \partial s^j} K : \X\times \X \rightarrow \Lin(\Y)$:
% \[
% 	\dfrac{\partial^2}{\partial r^i \partial s^j} K (z,x) := \left.\dfrac{\partial^2}{\partial r^i \partial s^j} K(r,s) \right| _{r=z,s=x},
% \]
% and the $C^0$ function $\partial^2_{ij} K_x y$ on $\Y^\X$:
% \[
% 	\partial^2_{ij} K_x y(z) := \dfrac{\partial} {\partial z^i} (\partial_j K_x y ) (z) = \dfrac{\partial^2}{\partial r^i \partial s^j} K (z ,x) y	\quad \forall z \in \X.
% \]
The following result provides a useful reproducing property for the derivatives of functions in $\Hk$ that will be instrumental in deriving the solution algorithm.

\begin{theorem}[Derivative Properties on $\Hk$ \cite{MicheliGlaunes2014}] \label{thm:RK_der}
Let $\Hk$ be a RKHS in $\Y(\X)$ with reproducing kernel $K \in C^2(\X \times \X)$. Then, for all $(x,y) \in \X \times \Y$:
\begin{enumerate}
	\item[(a)] $\partial_j K_x y \in \Hk$ for all $j = 1,\ldots,n$. 
	\item[(b)] The following derivative reproducing property holds for all $j = 1,\ldots, n$:
	\[
		\ip{\dfrac{\partial f(x)}{\partial x^j}}{y} = \ip{f}{ \partial_j K_x y}_{\Hk}, \quad \forall f \in \Hk.
	\]
\end{enumerate}
\end{theorem}

As mentioned in Section~\ref{sec:soln}, the bi-linearity of constraint~\eqref{nat_contraction_W} forces us to adopt an alternating solution strategy whereby in the ``dynamics" sub-problem, $\{W,\wl,\wu\}$ are held fixed and we minimize $J_d$ with respect to $\{f,B,\lambda\}$. In the ``metric" sub-problem, $\{f,B,\lambda\}$ are held fixed and we minimize $J_m$ with respect to $\{W,\wl,\wu\}$.

In the following we derive several useful representer theorems to characterize the solution of the two sub-problems, under the two simplifying assumptions introduced in Section~\ref{sec:finite}. 

\subsection{Dynamics Sub-Problem}\label{sec:dyn_sub}

Let $K^f$ be the reproducing $\mathcal{C}^2$ kernel for an $\Y$-valued RKHS $\Hk^f$ and let $K^B$ be another reproducing kernel for an $\Y$-valued RKHS $\Hk^B$. Define the finite-dimensional subspaces:
\begin{align}
	\Vf &:= \left\{ \sum_{i=1}^{N_c} K^f_{\xs} a_i + \sum_{i=1}^{N_c} \sum_{p=1}^{n} \partial_p K^f_{\xs} a'_{ip},\quad a_i, a'_{ip} \in \reals^n \right\} \subset \Hk^f. \label{rep} \\
	\Vb &:= \left\{ \sum_{i=1}^{N_c} \Kb_{\xs} c_i + \sum_{i=1}^{N_c} \sum_{p=1}^{n} \partial_p \Kb_{\xs} c'_{ip},\quad c_i, c'_{ip} \in \reals^n \right\} \subset \Hk^B.\label{rep_b}
\end{align}
Note that all $\xs$ taken from the training dataset of $(\xs,\us,\dot{x}_i)$ tuples are a subset of $X_c$.

\begin{theorem}[Representer Theorem for $f,B$]\label{thm:rep_dyn}
Suppose the reproducing kernel $K^B$ is chosen such that all functions $g \in \Hkb$ satisfy the sparsity structure $g^j(x) = 0$ for $j = 1,\ldots,n-m$. Consider then the pointwise-relaxed dynamics sub-problem for problem~\eqref{prob_gen2}:
\begin{align}
&\min_{\substack{\hat{f} \in \Hk^{f}, \ \hat{b}_j \in \Hk^{B}, j =1,\ldots,m \\ \lambda \in \reals_{>0}}} && J_d(\hat{f},\hat{B})  \nonumber \\
&\qquad \text{s.t.} && F(\xs; \hat{f},W,\lambda) \preceq 0, \quad \forall \xs \in X_c. \label{lmi_pw1}
\end{align}
Suppose the feasible set for the LMI constraint is non-empty. Denote $f^*$ and $b_j^*, j = 1,\ldots, m$ as the optimizers for this sub-problem. Then, $f^* \in \Vf$ and $b_j \in \Vb$ for all $j = 1,\ldots, m$. 
\end{theorem}
\begin{proof}
For a fixed $W$, the constraint is convex in $\hat{f}$ by linearity of the matrix $F$ in $\hat{f}$ and $\partial \hat{f}/\partial x$. By assumption~\eqref{ass:B_simp}, and the simplifying form for $B_{\perp}$, the matrix $F$ is additionally independent of $B$. Then, by strict convexity of the objective functional, there exists unique minimizers $f^*, b_j^* \in \Hk$, provided the feasible region is non-empty~\cite{KurdilaZabarankin2006}. 

Since $\Vf$ and $\Vb$ are closed, by the Projection theorem, $\Hk^f = \Vf \oplus \Vf^{\perp}$ and $\Hkb = \Vb \oplus \Vb^{\perp}$. Thus, any $g \in \Hk^f$ may be written as $g^{\Vf} + g^{\perp}$ where $g^{\Vf} \in \Vf$, $g^{\perp} \in \Vf^\perp$, and $\ip{g^{\Vf}}{g^{\perp}}_{\Hk^f} = 0$. Similarly, any $g \in \Hkb$ may be written as $g^{\Vb} + g^{\perp}$ where $g^{\Vb} \in \Vb$, $g^{\perp} \in \Vb^\perp$, and $\ip{g^{\Vb}}{g^{\perp}}_{\Hkb} = 0$.

Let $f = f^{\Vf} + f^{\perp}$ and $b_j = b_{j}^{\Vb} + b_{j}^\perp$, and define 
\[
	\begin{split}
	H^{\V}(x,u) &= f^{\Vf}(x) + \sum_{j=1}^{m} u^j b_j^{\Vb} (x) \\ 
	H^{\perp}(x,u) &= f^{\perp}(x) + \sum_{j=1}^{m} u^j b_j ^\perp (x).
	\end{split}
\]
We can re-write $J_d(f,B)$ as:
\[
\begin{split}
	J_d(f,B) &= \sum_{i=1}^{N} \ip{H^{\V} (x_i, u_i) - \dot{x}_i}{ H{^\V} (x_i,u_i) - \dot{x}_i} + 2\ip{H^{\V} (x_i, u_i) - \dot{x}_i}{H^\perp (x_i,u_i)}  \\
	&\qquad \qquad + \ip{H^\perp(x_i, u_i)}{H^\perp(x_i, u_i)} +  \mu_f \left( \|f^{\Vf}\|_{\Hk^f}^2 + \|f^\perp \|_{\Hk^f}^2 \right)  \\
	&\qquad \qquad + \mu_b \left( \sum_{j=1}^{n} \|b_j^{\Vb} \|_{\Hkb}^2 + \|b_j^\perp \|_{\Hkb}^2 \right). \\
\end{split}
\]
Now, 
\[
	\begin{split}
	&\ip{H^{\V} (\xs, u_i) - \dot{x}_i}{H^\perp (\xs,u_i)}  \\
				&\qquad = \ip{f^\perp (\xs) + \sum_{j=1}^{m} u_i^j b_j^\perp(\xs)}{H^\V (\xs, u_i) - \dot{x}_i} \\
				&\qquad = \underbrace{\ip{f^\perp}{K_{\xs}^f \left(H^\V (\xs, u_i) - \dot{x}_i \right)}_{\Hk}}_{=0} + \underbrace{\ip{\sum_{j=1}^{m} u_i^j b_j^\perp}{\Kb_{\xs} \left(H^\V (\xs, u_i) - \dot{x}_i \right)}_{\Hkb}}_{=0} 
	\end{split}
\]
since $K_{\xs}^f \left(H^\V (\xs, u_i) - \dot{x}_i \right) \in \Vf$, $\Vb^\perp$ is closed under addition, and $\Kb_{\xs} \left(H^\V (\xs, u_i) - \dot{x}_i \right) \in \Vb$. Thus, $J_d(f,B)$ simplifies to
\begin{equation}
\begin{split}
	\sum_{i=1}^{N} \left\| H^{\V} (\xs, u_i) - \dot{x}_i \right\|_2^2 &+ \mu_f  \|f^{\Vf}\|_{\Hk^f}^2 + \mu_b \sum_{j=1}^{n} \|b_j^{\Vb} \|_{\Hk^B}^2  + \\
	&+\left[ \sum_{i=1}^{N}  \left\| H^\perp (\xs, u_i) \right\|_2^2 +  \mu_f \|f^\perp \|_{\Hk^f}^2  + \mu_b\sum_{j=1}^{n}  \|b_j^\perp \|_{\Hk^B}^2 \right].
\end{split}
\label{obj_simp}
\end{equation}
Now, the $(p,q)$ element of $\partial_f W(\xs)$ takes the form 
\[
	\ip{ \frac{ \partial w_{pq} (\xs)}{\partial x}}{f (\xs)} = \ip{f}{K_{\xs}^f\frac{ \partial w_{pq} (\xs)}{\partial x}}_{\Hk^f} =  \ip{f^{\Vf}}{K^f_{\xs}\frac{ \partial w_{pq} (\xs)}{\partial x}}_{\Hk^f}.
\]
Column $p$ of $\frac{\partial f(\xs)}{\partial x} W(\xs)$ takes the form
\[
	\begin{split}
	\sum_{j=1}^{n} w_{j p} (\xs) \dfrac{\partial f(\xs)}{\partial x^j}  = \begin{bmatrix} \sum_{j=1}^{n}  w_{jp}(\xs) \ip{ \frac{\partial f(\xs)}{\partial x^j} }{ e_1} \\ \vdots \\ 	\sum_{j=1}^{n}  w_{jp}(\xs) \ip{ \frac{\partial f(\xs)}{\partial x^j} }{ e_n}	\end{bmatrix} & = 
												       	 \begin{bmatrix} \sum_{j=1}^{n}  w_{jp}(\xs) \ip{ f }{\partial_j K^f_{\xs} e_1}_{\Hk^f} \\ \vdots \\ 	\sum_{j=1}^{n}  w_{jp}(\xs) \ip{ f }{\partial_j K^f_{\xs} e_n}_{\Hk^f}	\end{bmatrix} \\
					&=  \begin{bmatrix} \sum_{j=1}^{n}  w_{jp}(\xs) \ip{ f^{\Vf} }{\partial_j K^f_{\xs} e_1}_{\Hk^f} \\ \vdots \\ 	\sum_{j=1}^{n}  w_{jp}(\xs) \ip{ f^{\Vf} }{\partial_j K^f_{\xs} e_n}_{\Hk^f}	\end{bmatrix},
	\end{split}												       	 
\]
where $e_i$ is the $i^{\text{th}}$ standard basis vector in $\reals^n$. Thus, $f^\perp$ plays no role in pointwise relaxation of constraint~\eqref{nat_contraction_W} and thus does not affect problem feasibility.
Given assumption~\ref{ass:B_simp}, $b^{\perp}$ also has no effect on problem feasibility. Thus, by non-negativity of the term in the square brackets in~\eqref{obj_simp}, we have that the optimal $f$ lies in $\Vf$ and optimal $b_j$ lies in $\Vb$. \qed
\end{proof}

The key consequence of this theorem is the reduction of the infinite-dimensional search problem for the functions $f$ and $b_j, j=1,\ldots,m$ to a finite-dimensional \emph{convex} optimization problem for the constant vectors $a_i,a'_{ip},\{c_i^{(j)}, c_{ip}^{(j)'}\}_{j=1}^{m} \in \reals^n$, by choosing the function classes $\mathcal{H}^f = \Hk^{f}$ and $\mathcal{H}^B = \Hk^{B}$. Next, we characterize the optimal solution to the metric sub-problem.

\subsection{Metric Sub-Problem}\label{sec:met_sub}

By the simplifying assumption in Section~\ref{sec:B_simp}, constraint~\eqref{killing_A} requires that $W_{\perp}$ be only a function of the first $(n-m)$ components of $x$. Thus, define $\kw : \X \times \X \rightarrow \reals$ as a \emph{scalar} reproducing kernel with associated real-valued scalar RKHS $\Hw$. Additionally, define $\kwp : \X \times \X \rightarrow \reals$ as another scalar reproducing kernel with associated real-valued scalar RKHS $\Hwp$. In particular, $\kwp$ is only a function of the first $(n-m)$ components of $x \in \X$ in both arguments.  

For all $(p,q) \in \{1,\ldots,(n-m)\}$, we will take $w_{pq} \in \Hwp$ (corresponding to $W_{\perp}$) and all other entries of $W$ as functions in $\Hw$. Define the kernel derivative functions:
\[
    \partial_j \kw_{x}(z) := \left.\dfrac{\partial \kw}{\partial r^j} \kw(r,s)\right|_{(r=x,s=z)} \quad \partial_j \kwp_{x}(z) := \left.\dfrac{\partial \kwp}{\partial r^j} \kw(r,s)\right|_{(r=x,s=z)} \quad \forall z \in \X.
\]
From~\cite{Zhou2008}, it follows that the kernel derivative functions satisfy the following two properties, similar to Theorem~\ref{thm:RK_der}:
\[
    \begin{split}
            \partial_j \kw_{x} \in \Hw \quad \forall j = 1,\ldots,n,\  x \in \X \\
            \dfrac{\partial f}{\partial x^j}(x) = \ip{f}{\partial_j \kw_x}_{\Hw}, \quad \forall f \in \Hw.
    \end{split}
\]
A similar property holds for $\kwp$ and $\Hwp$. Consider then the following finite-dimensional spaces:
\begin{align}
	\V_{\kw} &:= \left\{ \sum_{i=1}^{N_c}  a_i \kw_{\xs} + \sum_{i=1}^{N_c} \sum_{p=1}^{n}  a_{ip}' \,  \partial_p \kw_{\xs} ,\quad a_i, a_{ip}' \in \reals \right\} \subset \Hw \\
	\V_{\kwp} & := \left\{ \sum_{i=1}^{N_c} c_i \kwp_{\xs} + \sum_{i=1}^{N_c} \sum_{p=1}^{n}  c_{ip}' \,  \partial_p \kwp_{\xs} ,\quad c_i, c_{ip}' \in \reals \right\} \subset \Hwp 
\end{align}
and the proposed representation for $W(x)$:
\begin{equation}
	\begin{split}
	W(x) =  &\sum_{i=1}^{N_c} \hat{\Theta}_i \kwp_{\xs} (x) + \sum_{i=1}^{N_c}\sum_{j=1}^{n} \hat{\Theta}'_{ij} \partial_{j}\kwp_{\xs} (x)  \\
		+  & \sum_{i=1}^{N_c} \Theta_i \kw_{\xs} (x) + \sum_{i=1}^{N_c}\sum_{j=1}^{n} \Theta'_{ij} \partial_{j}\kw_{\xs} (x),
	\end{split}
\label{W_rep}
\end{equation}
where $\hat{\Theta}, \hat{\Theta}' \in \Sj_{n}$ are constant symmetric matrices with non-zero entries only in the top-left $(n-m)\times(n-m)$ block, and $\Theta,\Theta' \in \Sj_{n}$ are constant symmetric matrices with zero entries in the top-left $(n-m)\times(n-m)$ block.

\begin{theorem}[Representer Theorem for $W$]\label{thm:rep_W}
Consider the pointwise-relaxed metric sub-problem for problem~\eqref{prob_gen2}:
\begin{align}
&\min_{\substack{w_{pq} \in \Hwp, (p,q) \in \{1,\ldots,(n-m)\} \\ w_{pq} \in \Hw \text{ else} \\ \wl,\wu \in \reals_{>0}}} && J_m(W,\wl,\wu)  \nonumber \\
&\qquad \text{subject to} && F(\xs; \hat{f}, W, \lambda) \preceq 0, \quad \forall \xs \in X_c, \label{lmi_pw2} \\
&\qquad && \wl I_n \preceq W(\xs) \preceq \wu I_n, \quad \forall \xs \in X_c, \label{lmi_pw3}
\end{align}
Suppose the feasible set of the above LMI constrains is non-empty. Denote $W^*$ as the optimizer for this sub-problem. Then, $W^*$ takes the form given in~\eqref{W_rep}.
\end{theorem}
\begin{proof}
Notice that while the regularizer term is strictly convex, the surrogate loss function for the condition number is affine. However, provided the feasible set is non-empty, there still exists a minimizer (possible non-unique) for the above sub-problem. 

Now, notice that the $(p,q)$ element of $\partial_{f} W_{\perp} (\xs)$ takes the form
\[
	\begin{split} 
		\partial_{f} w_{\perp_{pq}}(\xs) &=\sum_{j=1}^{n} \dfrac{\partial w_{\perp_{pq}} (\xs)}{\partial x^j} f^j (\xs) \\
								&= \sum_{j=1}^{n} \ip{w_{\perp_{pq}}}{\partial_{j}\kwp_{\xs}}_{\Hwp} f^j(\xs) \\
								&= \ip{w_{\perp_{pq}}}{ \sum_{j=1}^{n} f^j(\xs) \partial_{j}\kwp_{\xs}}_{\Hwp} \\
								&= \ip{w^{\V_{\kwp}}_{\perp_{pq}}}{ \sum_{j=1}^{n} f^j(\xs) \partial_{j}\kwp_{\xs}}_{\Hwp} 
	\end{split}
\]
Additionally, the $(p,q)$ element of $W(\xs)$ takes the form
\[
		w_{pq}(\xs) = \begin{cases}  \ip{w^{\V_{\kwp}}_{pq}}{\kwp_{\xs}}_{\Hwp} &\text{ if } (p,q) \in \{1,\ldots,(n-m)\} \\
						 	    \ip{w^{\V_\kw}_{pq}}{\kw_{\xs}}_{\Hw} &\text{ else.}
					\end{cases}
\]
That is, constraints~\eqref{nat_contraction_W} and uniform definiteness at the constraint points in $X_c$ can be written in terms of functions in $\V_\kw$ and $\V_{\kwp}$ alone. By strict convexity of the regularizer, and recognizing that $W(x)$ is symmetric, the result follows. \qed 
\end{proof}

Similar to the previous section, the key property here is the reduction of the infinite-dimensional search over $W(x)$ to the finite-dimensional convex optimization problem over the constant symmetric matrices $\Theta_i,\Theta_{ij}',\hat{\Theta}_i,\hat{\Theta}'_{ij}$ by choosing the function class for the entries of $W$ using the scalar-valued RKHS.

At this point, both sub-problems are finite-dimensional convex optimization problems. Crucially, the only simplifications made are those given in Section~\ref{sec:finite}. However, a final computational challenge here is that the number of parameters scales with the number of training $N$ and constraint $N_c$ points. This is a fundamental challenge in all non-parametric methods. In the next section we present the dimensionality reduction techniques used to alleviate these issues.

\subsection{Approximation via Random Matrix Features}
\label{sec:feat}

The size of the problem using full matrix-valued kernel expansions in grows rapidly in $N_c n$, the number of constraint points times the state dimensionality. This makes training slow for even moderately long demonstrations even in low-dimensional settings. The induced dynamical system is slow to evaluate and integrate at inference time. Random feature approximations to kernel functions have been extensively used to scale up training complexity and inference speed of kernel methods~\cite{RahimiRecht2007, HuangAvronEtAl2014} in a number of applications.  The quality of approximation can be explicitly controlled by the number of random features. In particular, it has been shown~\cite{RahimiRecht2008} that any function in the RKHS associated with the exact kernel can be approximated to arbitrary accuracy by a linear combination of sufficiently large number of random features. 

These approximations have only recently been extended to matrix-valued kernels~\cite{Minh2016,BraultHeinonenEtAl2016}. Given a matrix-valued kernel $K$, one defines an appropriate matrix-valued feature map $\Phi : \X \rightarrow \reals^{d\times n}$ with the property 
\[
    K(x,z) \approx \Phi(x)^T \Phi(z),
\]
where $d$ controls the quality of this approximation. With such an approximation, notice that \emph{any} function $g$ in the associated RKHS $\Hk$ can be re-parameterized as:
\[
    g(x) = \sum_{i=1}^{N_c} K(x,\xs) a_i \approx \sum_{i=1}^{N_c} \Phi(x)^T \Phi(\xs) a_i = \Phi(x)^T \alpha,
\]
where $\alpha = \sum_{i=1}^{N_c} \Phi(\xs)a_i \in \reals^d$. Thus, the function $g$ is now summarized by the $d-$ dimensional vector $\alpha$ instead of $N_c$ vectors in $\Y$ ($N_c n$ parameters). 
Applying a similar trick to:
\[
    \dfrac{\partial}{\partial s^j} K(x,z) \approx \left. \Phi(x)^T \dfrac{\partial}{\partial s^j} \Phi(s) \right|_{s =z},
\]
one can approximate terms of the form
\[
    \sum_{i=1}^{N_c}\sum_{j=1}^{n} \partial_j K_{\xs} c_{ij} (x)
\]
by $\Phi(x)^T v$ for some $v \in \reals^d$. Applying this approximation to the subspaces in~\eqref{rep},~\eqref{rep_b}, and~\eqref{W_rep}, we finally arrive at the function parameterizations\footnote{To apply this decomposition to $W(x)$, we simply leverage vector-valued feature maps for each entry of $W(x)$.} in eqs.~\eqref{param_1}--\eqref{param_2}, and the constraints reformulation in problem~\eqref{learn_finite} in terms of the constant vectors. The regularization terms in the objective follow from the definition of the inner product in $\Hk$:
\small{
\[ 
\begin{split}
    \|g\|_{\Hk}^2 = \sum_{i,j=1}^{N_c} \ip{a_i}{K(x_i,x_j)a_j} \approx \sum_{i,j=1}^{N_c} \ip{\Phi(x_i)a_i}{\Phi(x_j)a_j} &= \ip{\sum_{i=1}^{N_c}\Phi(x_i)a_i}{\sum_{j=1}^{N_c}\Phi(x_j)a_j} \\
    &= \|\alpha\|_2.
\end{split}
\]
}\normalsize
A canonical example is the the Gaussian separable kernel $K_{\sigma}(x,z) := e^{\frac{-\|x-z\|_2^2}{\sigma^2}} I_n$ with feature map:
\[
    \dfrac{1}{\sqrt{s}}\begin{bmatrix} \cos(\omega_1^T x) \\ \sin(\omega_1^T x) \\ \vdots \\ \cos(\omega_s^T x) \\ \sin(\omega_s^T x)    \end{bmatrix} \otimes I_n 
\]
where $\omega_1,\ldots,\omega_s$ are i.i.d. draws from $\mathcal{N}(0,\sigma^{-2}I_n)$. 

\section{Solution Parameters for PVTOL}
\label{app:prob_params}

Notice that the true input matrix for the PVTOL system satisfies Assumption~\ref{ass:B_simp}. Furthermore, it is a constant matrix. Thus, the feature mapping $\Phi_b$ is therefore just a constant matrix with the necessary sparsity structure.

The feature matrix for $f$ was generated using the Random Fourier approximation to the Gaussian separable matrix-valued kernel (see Appendix~\ref{sec:feat}) with $\sigma = 6$ and $s = 8n = 48$ sampled Gaussian directions, yielding a feature mapping matrix $\Phi_f$ with $d_f = 576$ (96 features for each component of $f$). The scalar-valued reproducing kernels for the entries of $W$ were taken to be the Gaussian kernel with $\sigma = 15$. To satisfy condition~\eqref{killing_A}, the kernel for $w_{ij}, \ (i,j) \in \{1,\ldots,(n-m)\}$ was only a function of the first $n-m$ components of $x$. A total of $s = 36$ Gaussian samples were taken to yield feature vectors $\phi_w$ and $\hat{\phi}_w$ of dimension $d_w = 72$.  Furthermore, by symmetry of $W(x)$ only $n(n+1)/2$ functions were actually parameterized. Thus, the learning problem in~\eqref{learn_finite} comprised of $d_f + d_w n(n+1)/2 + 4 = 508$ parameters for the functions, plus the extra scalar constants $\lambda,\wl,\wu$.

The learning parameters used were: model N-R (all $N$): $\mu_f = 0, \mu_b = 10^{-6}$, R-R (all $N$): $\mu_f = 10^{-6}, \mu_b = 10^{-6}$, CCM-R (all $N$): $\mu_f = 10^{-3}, \mu_b = 10^{-6}$; $N \in \{100,250,500\}: \mu_w = 10^{-3}$, $N = 1000: \mu_w = 10^{-4}$. Tolerance parameters: constraints: $\{\delta_{\lambda},\epsilon_{\lambda},\delta_{\wl}, \epsilon_{\wl}\} = \{0.01,0.01,0.01,0.01\}$; discard tolerance $\delta = 0.05$. Note that a small penalization on $\mu_b$ was necessary for all models due to the fact that feature matrix $\Phi_b$ is rank deficient.

\section{CCM Controller Synthesis} \label{ccm_appendix}

Let $\Gamma(p,q)$ be the set of smooth curves $c:[0,1] \rightarrow \X$ satisfying $c(0) = p, c(1) = q$, and define $\delta_c(s) := \partial c(s)/\partial s$. At each time $t$, given the nominal state/control pair $(x^*,u^*)$ and current actual state $x$:
\begin{leftbox}
\begin{itemize}[leftmargin=.4in]
    \item[{\bf Step 1:}] Compute a curve $\gamma \in \Gamma(x^*,x)$ defined by:
\begin{equation}
    \gamma \in \argmin_{c\in \Gamma(x^*, x)} \int_{0}^{1} \delta_c(s)^T M(c(s)) \delta_c (s) ds,
\end{equation}
    and let $\mathcal{E}$ denote the minimal value.
    \item[{\bf Step 2:}] Define: 
    \begin{equation}
    \begin{split}
        \mathcal{E}_d(k):= &2 \delta_{\gamma}(1)^T M(x) (f(x) + B(x) (u^* + k)) \\
        &- 2\delta_{\gamma}(0)^T M(x^*) (f(x^*) + B(x^*)u^*).
    \end{split}
    \label{ccm_ineq}
    \end{equation}
    \item[{\bf Step 3:}] Choose $k(x^*,x)$ to be any element of the set:
    \begin{equation}
        \mathcal{K} := \left\{ k : \mathcal{E}_d(k) \leq -2\lambda \mathcal{E} \right\}.
    \label{ccm_cntrl}
    \end{equation}
\end{itemize}
\end{leftbox}

By existence of the metric/differential controller pair $(M,\delta_u)$, the set $\mathcal{K}$ is always non-empty. The resulting $k(x^*,x)$ then ensures that the solution $x(t)$ indeed converges towards $x^*(t)$ exponentially~\cite{SinghMajumdarEtAl2017}.

From an implementation perspective, note that having obtained the curve $\gamma$, constraint~\eqref{ccm_cntrl} is simply a linear inequality in $k$. Thus, one can analytically compute a feasible feedback control value, e.g., by searching for the smallest (in any norm) element of $\mathcal{K}$; for additional details, we refer the reader to~\cite{ManchesterSlotine2017,SinghMajumdarEtAl2017}.

\end{document}